\documentclass[12pt,3p]{article}

\usepackage[english]{babel}
\usepackage{float}
\usepackage{graphicx,epstopdf}
\usepackage{amsmath,amssymb,amsfonts,mathrsfs} 
\usepackage{multicol,subcaption}
\usepackage{comment}
\usepackage{cite}
\usepackage{adjustbox}
\usepackage{pdflscape}
\usepackage{array}
\usepackage{booktabs}
\usepackage{tabularx,ragged2e,booktabs}
\usepackage{amsthm}
\usepackage{color}
\usepackage[linesnumbered,lined,ruled,vlined]{algorithm2e}
\newtheorem{corollary}{Corollary}

\newtheorem{observation}{Observation}
\newtheorem{proposition}{Proposition}

\usepackage{lineno}
\makeatletter
\def\makeLineNumberLeft{%
  \linenumberfont\llap{\hb@xt@\linenumberwidth{\LineNumber\hss}\hskip\linenumbersep}
  \hskip\columnwidth
  \rlap{\hskip\linenumbersep\hb@xt@\linenumberwidth{\hss\LineNumber}}\hss}
\leftlinenumbers
\makeatother

\begin{document}



\title{A Faster Method to Estimate Closeness Centrality Ranking}

%


\author{%
  Akrati Saxena*\\akrati.saxena@iitrpr.ac.in
  \and Ralucca Gera**\\rgera@nps.edu
  \and S. R. S. Iyengar*\\sudarshan@iitrpr.ac.in
  \and *Department of Computer Science and Engineering,\\ Indian Institute of Technology Ropar, India
  }
\date{%
  **Department of Applied Mathematics\\
Naval Postgraduate School, \\Monterey, CA 93943 USA
}

\maketitle

\begin{abstract}
Closeness centrality is one way of measuring how central a node is in the given network. The closeness centrality measure assigns a centrality value to each node based on its accessibility to the whole network. In real life applications, we are mainly interested in ranking nodes based on their centrality values. The classical method to compute the rank of a node first computes the closeness centrality of all nodes and then compares them to get its rank. Its time complexity is $O(n \cdot m + n)$, where $n$ represents total number of nodes, and $m$ represents total number of edges in the network. In the present work, we propose a heuristic method to fast estimate the closeness rank of a node in $O(\alpha \cdot m)$ time complexity, where $\alpha = 3$. We also propose an extended improved method using uniform sampling technique. This method better estimates the rank and it has the time complexity $O(\alpha \cdot m)$, where $\alpha \approx 10-100$. This is an excellent improvement over the classical centrality ranking method. The efficiency of the proposed methods is verified on real world scale-free social networks using absolute and weighted error functions.

\end{abstract}

\section{Introduction}


In network science, various centrality measures have been defined to identify important nodes based on the application context. As an example, one can consider identifying the location within a city where to place a new public service, so that it is easily accessible for everyone. Similarly, identifying central people that have ideal social network location for the purpose of information dissemination or network influence. In such kind of applications, the nodes who can access the entire network faster need to be selected. To capture this specific property of the reachability of nodes to the entire network, researchers have defined closeness centrality.

The closeness centrality of a node $u$ is defined as $C(u) = \frac{n-1}{\sum_{\forall v}d(u,v)}$, where $n$ is total number of nodes and $d(u,v)$ is the shortest distance between two nodes $u$ and $v$ \cite{freeman1978centrality}. It represents the closeness of a node to all other nodes in the given network.
The closeness centrality of a node can be computed by executing breadth first traversal (BFT) \cite{thomas2001introduction} from the respective node. The time complexity to compute the closeness centrality of a node is $O(m)$, where $m$ represents total number of edges in the network. The proposed definition of closeness centrality \cite{freeman1978centrality} is only applicable for the connected networks. In the present work, we have considered real world scale-free connected networks.

The proposed closeness centrality measure assigns a centrality value to each node. But in real life applications, we are mainly interested in the relative importance of the node rather than its centrality value. This can be measured using the closeness centrality rank of the node based on the computed closeness centrality values of all nodes. The classical method to compute the closeness rank of a node has two steps: ($1$) calculate the closeness centrality values of all nodes, and  ($2$) compare these values to determine the closeness rank of the node. The time complexity of the first step is $O(n \cdot m)$ as it computes the closeness centrality of all nodes, and for the second step, it is $O(n)$ as it compares the centrality value of the given node with all other nodes. So, the overall time complexity of this process is $O(n \cdot m +n)= O(n \cdot m)$ that is very high, given that the entire network is required. This method is not feasible for large real world networks because of its high time complexity and dynamic characteristics of the networks.

\begin{figure}[t]
\centering
\includegraphics[width=0.8\linewidth]{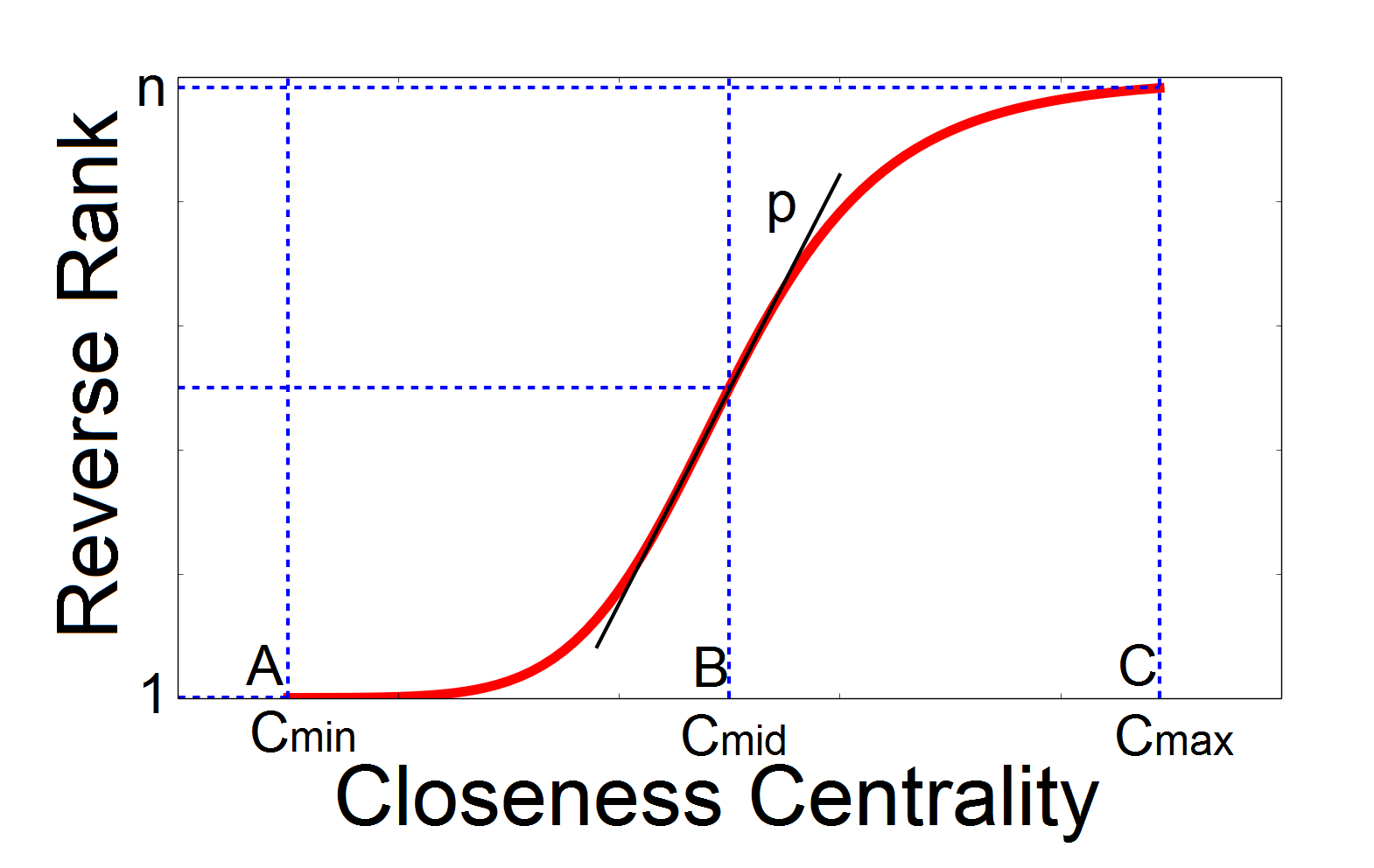}
\caption{Reverse Rank versus Closeness Centrality}
\label{fig:sorted}
\end{figure}

Real world complex networks such as WWW network, online social networks, collaboration networks, communications networks, are growing very fast with time. So, the computation of different centrality measures based on the global structure of the network is very high. Current literature focuses mainly on proposing approximation methods to compute closeness centrality measure. Even if we apply these approximation methods, we need to approximate closeness centrality values of all nodes to estimate the rank of a node, and its complexity is still high. Besides the size, there are other issues like real world networks are highly dynamic in nature. In these dynamic networks, the importance of different nodes keeps changing with time. So, to compute the rank of a node, we need the current complete snapshot of the network. Due to the large size of the network, it may not be feasible to download the entire network, store it, and process it.

In the present work, we propose heuristic and randomized heuristic methods to fast estimate the closeness centrality rank of a node. The complexity of the proposed methods is $O(m)$ that is a great improvement over the classical ranking method that has the complexity $O(nm)$. The proposed methods exploit the structural properties of the network to reduce the time complexity of closeness ranking. For example, the nodes in the center of a network have high closeness values, while the ones in the periphery have the least closeness values.
Closeness values of the middle layer nodes increase sharply from periphery to center. Due to this unique behavior of closeness centrality in real world scale-free social networks, we observe that the reverse ranking versus closeness centrality follows a sigmoid curve as shown in Figure~\ref{fig:sorted}. In reverse ranking, the node having the lowest closeness value has the highest rank $1$ and the node having the highest closeness value has the least rank.
This sigmoid curve helps to estimate the closeness rank of a node without computing the closeness values of all nodes.

The proposed methods are simulated on different types of real world networks, such as online social networks, collaboration network, communication network, and so on. The accuracy of the proposed methods is measured using absolute and weighted error functions. Results show that the proposed method can be used efficiently for large size dynamic networks.  

The main contributions of the paper are as follows:
\begin{itemize}
\item We study the characteristics of closeness centrality and their dependency on the structural properties of the real world scale-free social networks.
\item We propose methods to estimate the closeness centrality rank of a node in $O(m)$ time.
\item The proposed methods are simulated on real world networks and their efficiency is verified using absolute and weighted error functions.
\end{itemize}

As per the best of our knowledge, this is the first work in this direction.
The rest of this paper is organized as follows. In the next section, we discuss the brief literature on closeness centrality. Section \ref{sec-datasets} and \ref{sec-notation} explain datasets and notation used in the paper respectively. In section \ref{sec-ccb}, we study the behavior of closeness centrality. These observations help to construct the closeness ranking estimation method. In section \ref{sec-sam} and \ref{sec-ram}, we propose methods to estimate the closeness rank and discuss their complexity analysis. In section \ref{sec-results}, the simulation results are discussed. The paper is concluded in section \ref{sec-conclusion}. The proposed work further opens up various future directions that are discussed in the last section.

\section{Related Work}\label{relatedwork}

Closeness centrality denotes reachability of a node to the given network. In undirected and unweighted networks, the reachability of two nodes only considers the minimum number of hops to reach from one node to another. But in other types of networks like weighted networks, directed networks, the link weights and directions also affects the distance between two nodes.
So, the closeness centrality has been extended to these networks like weighted networks \cite{ruslan2015improved}, directed networks \cite{du2015new}, disconnected networks \cite{rochat2009closeness}, multilayer networks \cite{barzinpour2014clustering}, overlapped community structure \cite{tarkowski2016closeness}, and so on.

The complexity to compute the closeness centrality in large scale networks is very high. It has attracted the researchers towards the following problems related to measuring the closeness centrality of a node:
\begin{enumerate}
\item Update closeness centrality in dynamic networks,
\item Approximation algorithms for closeness centrality,
\item Identify top-$k$ nodes,
\item Other works like the computation of closeness centrality in parallel or distributed environment, its correlation with other centrality measures, and so on.
\end{enumerate}

Real world networks are highly dynamic and the computation of closeness centrality of all nodes for each change in the networks will be a cumbersome task. In dynamic networks, for each update, the closeness centrality of some nodes may remain unaffected. Kas et al. proposed a method to update closeness centrality in dynamic networks \cite{kas2013incremental}. The proposed method uses the set of affected nodes to update the closeness centrality whenever there is any addition, removal, or modification of nodes or edges. Yen also proposed an algorithm called CENDY (Closeness centrality and avErage path leNgth in DYnamic networks) to update closeness centrality whenever an edge is updated \cite{yen2013efficient}. Sariyuce et al. proposed a method to update closeness centrality using the level difference information of breadth first traversal \cite{sariyuce2013incremental}.

The classical method to compute closeness centrality of a node requires the entire network and it is costly for big networks. There are some works that provide approximation methods to compute it fast.
Cohen et al. proposed a sampling based method to approximate closeness centrality in directed and undirected networks \cite{dellingcomputing}. Eppstein et al. proposed a randomized approximation algorithm with time complexity $O(\frac{logn}{\epsilon^2} \cdot m)$ to approximate the closeness centrality within an additive error of $\epsilon \cdot \Delta$ with high probability, where $\Delta$ is the diameter of the network \cite{eppstein2004fast}. Rattigan used the concept of network structure index (NSI) to approximate the values of different centrality measures that are based on the shortest paths in the given network \cite{rattigan2006using}. Some other approximation methods for closeness centrality include \cite{chan2009fast, brandes2007centrality, pfeffer2012k}.

Most of the real life applications focus on identifying few top nodes having the highest closeness centrality. Okamoto et al. proposed a method to rank $k$ highest closeness centrality nodes using a hybrid of approximate and exact algorithms \cite{okamoto2008ranking}. Ufimtsev proposed an algorithm to identify high closeness centrality nodes using group testing \cite{ufimtsev2014extremely}. Olsen et al. presented an efficient technique to find $k$ most central nodes based on closeness centrality \cite{olsen2014efficient}. They used intermediate results of centrality computation to minimize the computation time. Bergamini et al. proposed a faster method to identify top-$k$ nodes in undirected networks by approximating the upper bound on closeness centrality using BFT \cite{bergaminicomputing}.

Wehmuth et al. studied the correlation of closeness centrality with the local neighborhood volume of the node \cite{wehmuth2012distributed}. The ranking based on local neighborhood volume is named as DACCER (Distributed Assessment of the Closeness CEntrality Ranking) and is highly correlated with closeness centrality ranking in both real world and synthetic networks.

Bader et al. proposed parallel algorithm to compute closeness centrality, where it executes a breadth first traversal (BFT) from each vertex as a root \cite{bader2006parallel}. Lehmann and Kaufmann proposed a method for decentralized computation of closeness centrality \cite{lehmann2003decentralized}. Wang et al. proposed a distributed algorithm that estimates closeness centrality with $91\%$ accuracy in terms of ordering on random geometric, Erd\H{o}s-R\'{e}nyi, and Barab\'{a}si-Albert graphs \cite{wang2015distributed}.

Closeness centrality has been applied to study various research topics like collaboration networks \cite{newman2001scientific, yan2009applying}, brain network \cite{sporns2007identification}, community detection \cite{jarukasemratana2014community}, identification of the community of a node by using the community information of other nodes \cite{zhang2012inferring}, closeness preferential attachment (CPN) model to generate synthetic networks \cite{ko2008rethinking}, and so on.

Currently, there is no work to estimate the rank of a node using its closeness centrality. In our previous works, we proposed methods to estimate the degree rank of a node using its local information. First, we proposed a method based on the power law degree distribution of scale-free networks and it computes the degree rank of a node in $O(1)$ time \cite{saxena2015rank, saxena2016estimating}. We further compute the variance in the rank estimation using power law degree distribution \cite{saxena2015estimating}. Next, we propose sampling based methods to estimate the degree rank in scale-free and random networks \cite{saxena2017degree}. These works focus on directly estimating the rank of a node based on the centrality measure. In the present work, we rank the nodes based on closeness centrality that itself is a global centrality measure.

\section{Datasets}\label{sec-datasets}

We have studied structural properties of the closeness centrality using real world network datasets that are briefly explained in Table \ref{datasets}.

\begin{table}[htp]
\centering
\caption{Datasets}
\label{datasets}
\begin{tabular}{|l|l|l|l|l|}
\hline
Network & Type & $\#$Nodes & $\#$Edges & Ref \\ \hline 
Brightkite & Social Network & 56739 & 212945 & \cite{cho2011friendship} \\ \hline
DBLP & Co-authorship Network & 317080 & 1049866 & \cite{yang2015defining} \\ \hline
Digg & Social Network & 261489 & 1536577 & \cite{hogg2012social} \\ \hline
Enron & Communication Network & 84384 & 295889 & \cite{klimt2004enron} \\ \hline
Epinion & Social Network & 75877 & 405739 & \cite{richardson2003trust} \\ \hline
Facebook & Social Network & 63392 & 816831 &  \cite{viswanath2009evolution} \\ \hline
Gowalla & Social Network & 196591 & 950327 & \cite{cho2011friendship} \\ \hline
Google+ & Social Network & 107614 & 12238285 & \cite{mcauley2012learning} \\ \hline
Slashdot & Social Network & 82168 & 504230 &  \cite{leskovec2010signed} \\ \hline
Twitter & Social Network & 81306 & 1342296 & \cite{ mcauley2012learning} \\ \hline
\end{tabular}
\end{table}

\section{Terminologies}\label{sec-notation}

Table \ref{notation} explains the terminologies used in the paper. $G(V, E)$ represents a network where $V$ is the set of nodes and $E$ is the set of edges.

\begin{table}[htp]
\centering
\caption{Terminologies for the paper}
\label{notation}
\resizebox{\columnwidth}{!}{%
\begin{tabular}{|l|l|}
\hline
\textbf{Notation} & \textbf{Description} \\ \hline
$n$ & Total number of nodes in the network \\ \hline
$m$ & Total number of edges in the network \\ \hline
$u,v,w$ & Nodes in the network  \\ \hline
$deg(u)$  & Degree of node $u$  \\ \hline
$C(u)$  & Closeness centrality of node $u$  \\ \hline
$c_{max}$  & Maximum closeness centrality in the network  \\ \hline
$c_{min}$  & Minimum closeness centrality in the network  \\ \hline
$c_{mid}$  & Closeness centrality of middle ranked node in the network  \\ \hline
$c'_{max}$  & Estimated maximum closeness centrality in the network  \\ \hline
$c'_{min}$  & Estimated minimum closeness centrality in the network  \\ \hline
$c'_{mid}$  & Estimated closeness centrality of middle ranked node \\ \hline
$R_{rev}(u)$  & Reverse rank of node $u$ in the network \\ \hline
$R_{act}(u)$  & Actual rank of node $u$ in the network \\ \hline
$R_{est}(u)$  & Estimated rank of node $u$ in the network  \\ \hline
\end{tabular}
}
\end{table}

\section{Closeness Centrality Behavior}\label{sec-ccb}


In this section, we study the behavioral characteristics of closeness centrality and their correlation with the network structure.

\begin{figure*}[htp]
  \centering
  \subcaptionbox{Brightkite}[.45\linewidth][c]{%
    \includegraphics[width=.45\linewidth]{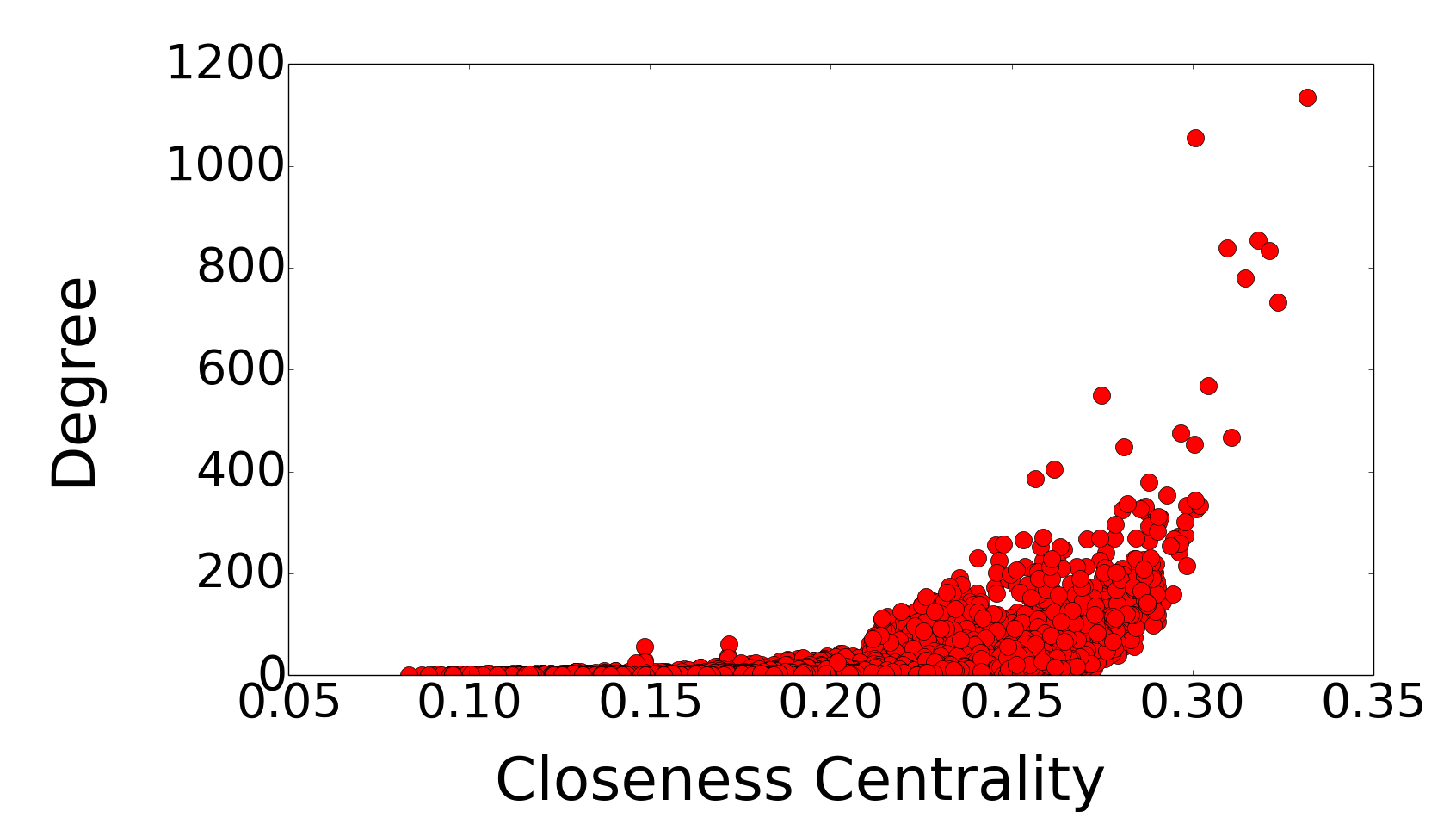}}\quad
  \subcaptionbox{DBLP}[.45\linewidth][c]{%
    \includegraphics[width=.45\linewidth]{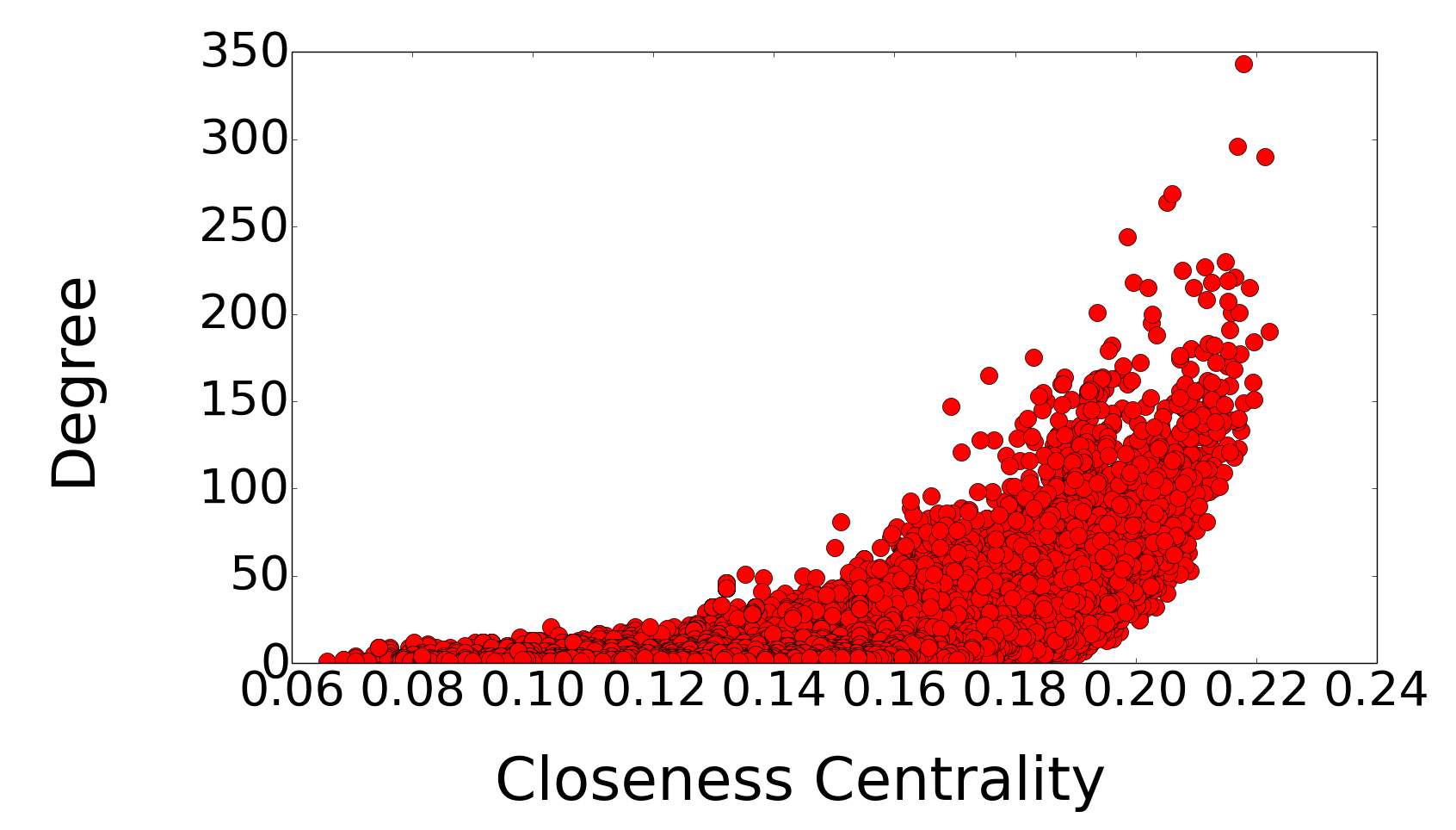}}\quad
  \subcaptionbox{Digg}[.45\linewidth][c]{%
    \includegraphics[width=.45\linewidth]{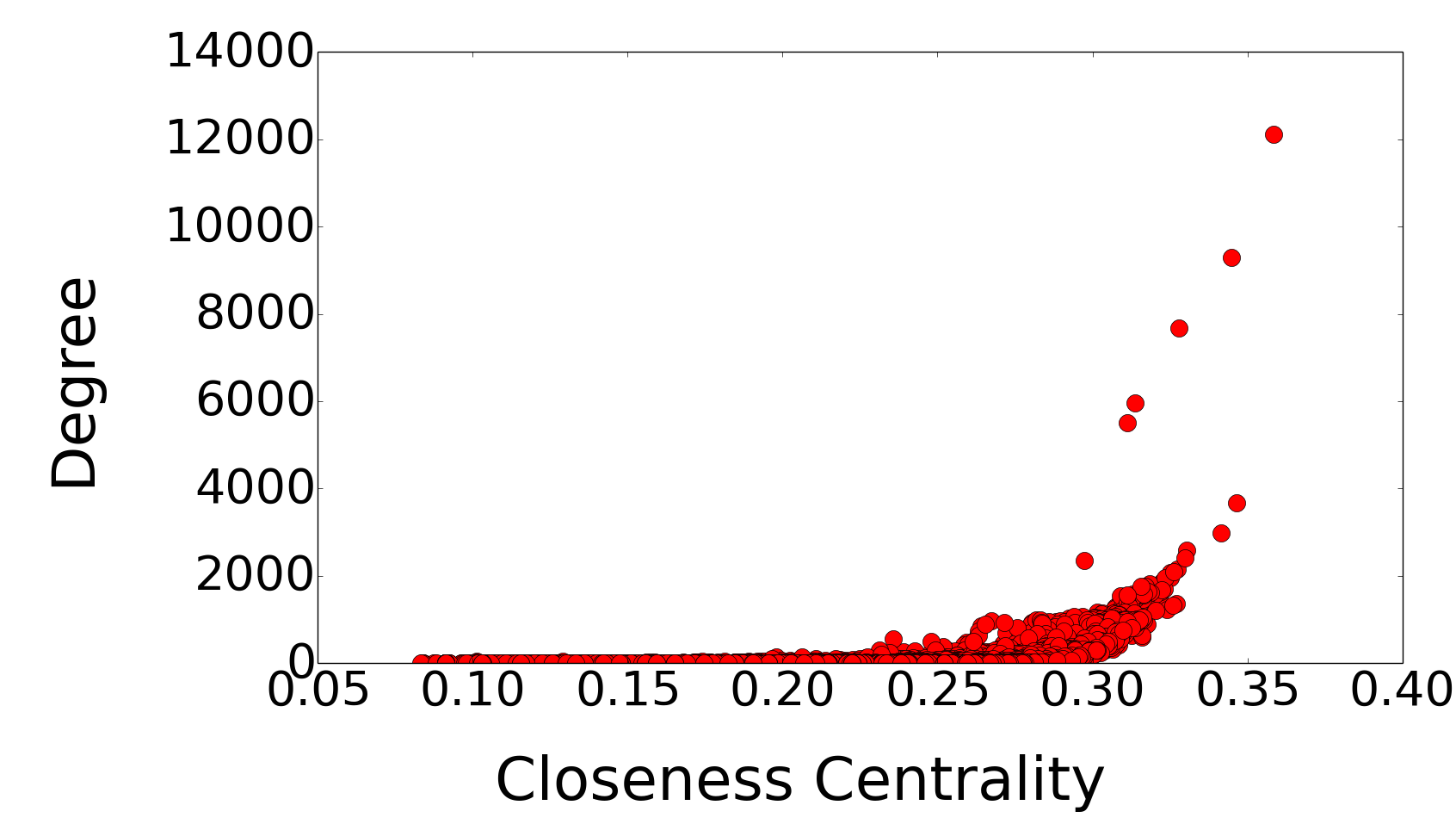}}\quad
  \subcaptionbox{Enron}[.45\linewidth][c]{%
    \includegraphics[width=.45\linewidth]{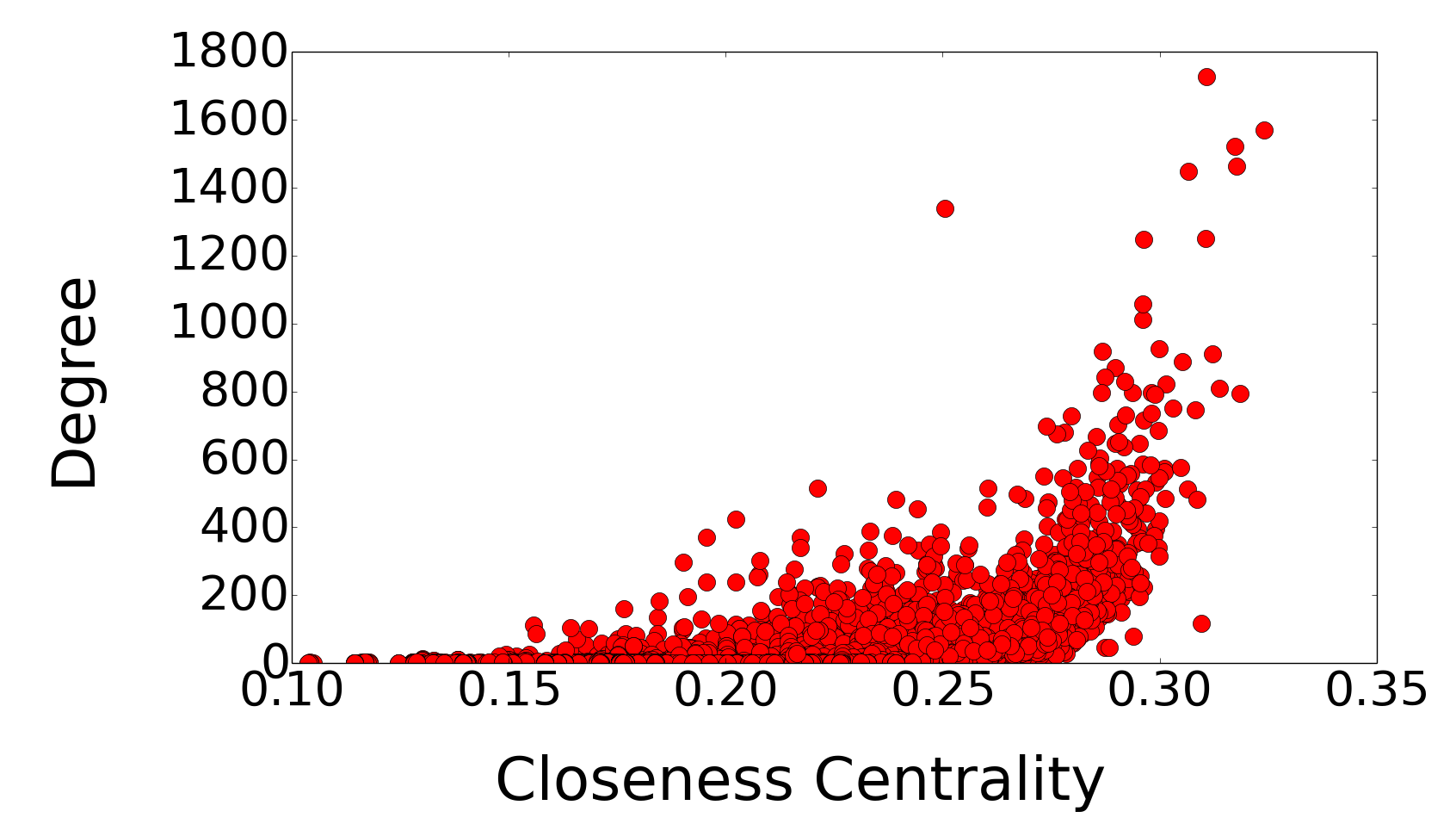}}\quad
  \subcaptionbox{Epinion}[.45\linewidth][c]{%
    \includegraphics[width=.45\linewidth]{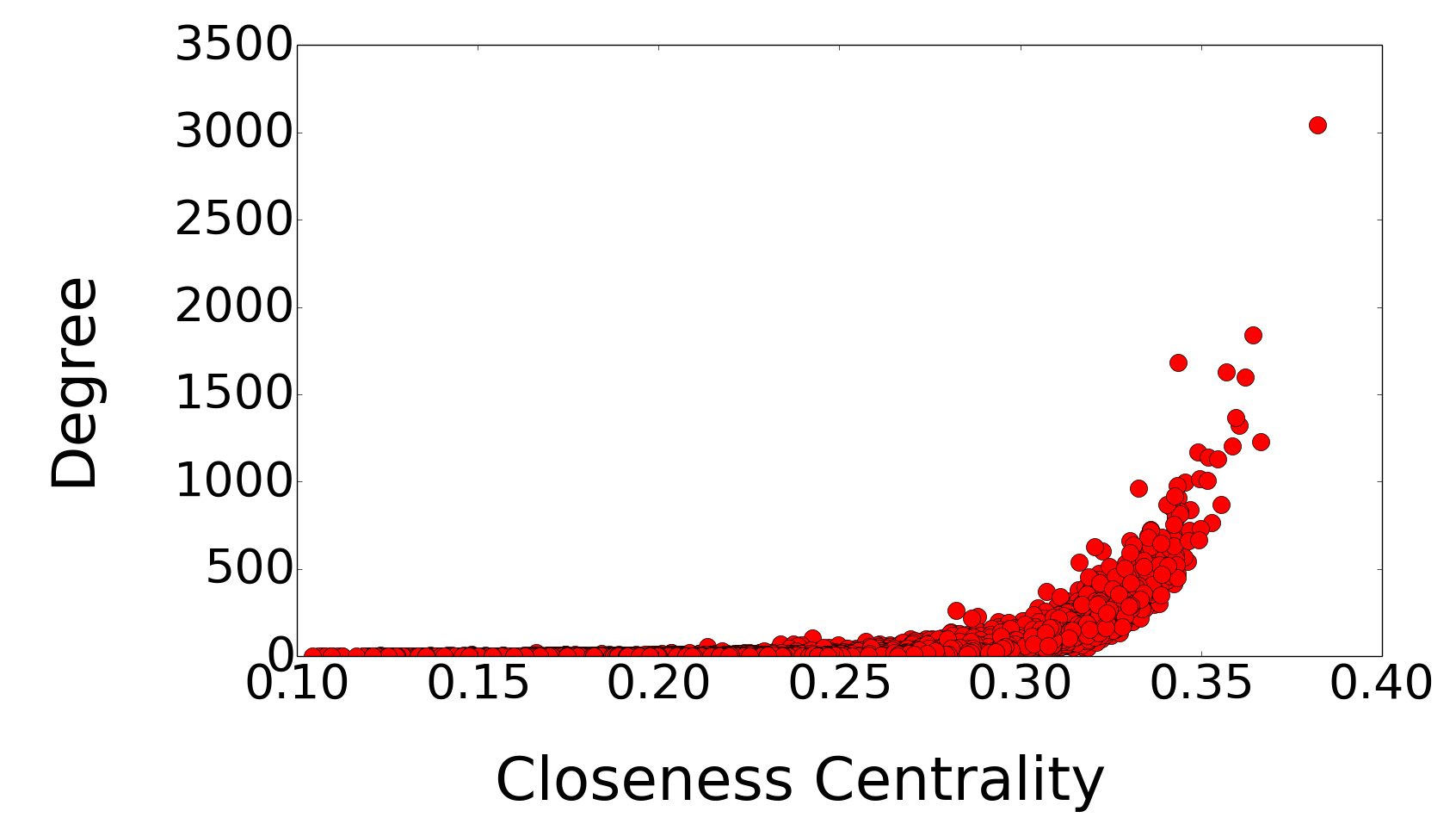}}\quad
  \subcaptionbox{Facebook}[.45\linewidth][c]{%
    \includegraphics[width=.45\linewidth]{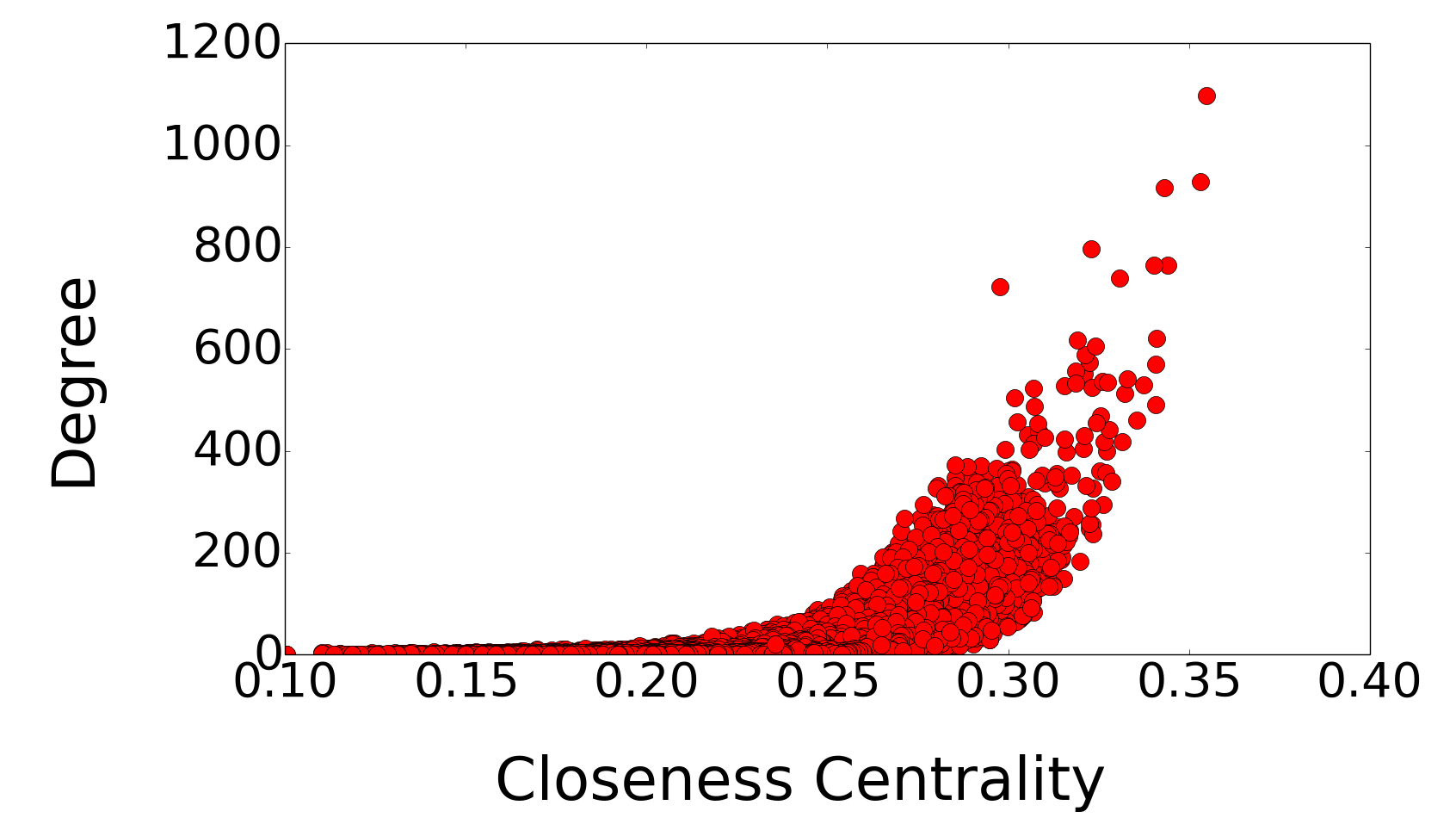}}\quad
  \caption{Degree versus Closeness Centrality}
  \label{fig-deg}
\end{figure*}

\subsection{Closeness Centrality vs. Degree Centrality}\label{sec-closedeg}

First, we study the correlation between closeness centrality and degree of the nodes. Results do not show that the closeness centrality is correlated with the degree, but they show that the node having the highest degree either has the highest closeness centrality or it is very close to the highest closeness centrality as shown in Figure~\ref{fig-deg}. This is desirable since the degree of a node is a local characteristic and it can be computed in $O(1)$ time. This information can be used to identify the node having highest closeness centrality. Further details are explained in Section~\ref{sec-sam}.

\subsection{Closeness Centrality Pattern from Center to Periphery}\label{sec-bft}

In scale-free networks, central nodes have the highest closeness centrality and the extreme periphery nodes have minimum closeness centrality. As we move from the center to the periphery closeness centrality of the nodes decreases. To analyze this pattern in depth, we execute breadth first traversal from the central node (having the highest closeness centrality) until all nodes are traversed. The plots of closeness centrality of nodes versus their distance from the center node are shown in Figure~\ref{fig-bft}. Results show that the nodes farthest from the central nodes have minimum closeness centrality, and the outermost level of the BFT (also referred as the outermost periphery) is very sparse. This behavior of the closeness centrality can be used to find the minimum closeness centrality nodes in the network, by looking at the diametrically opposed nodes from the central ones. 


\subsection{Closeness Centrality vs. Reverse Ranking}\label{sec-revrank}

Real world networks have a dense central region that contains a very few number of nodes having the highest closeness centrality. These central nodes are highly connected with each other and also with rest of the network. The nodes belonging to the outermost peripheral layers have the least closeness values. Closeness centrality of all other nodes lies between this range and increases sharply as we move from the periphery to the center.

Due to this behavior, the reverse ranking versus closeness centrality of nodes follows a sigmoid curve as shown in Figure~\ref{fig:sorted}.
In reverse ranking, a node having highest closeness value will have the smallest rank $n$ (where $n$ is the total number of nodes) and the node having the lowest closeness value will have the rank $1$.
We plot reverse rank versus closeness centrality for more than $20$ real world social networks and find that they follow the sigmoid curve and it is symmetric in most of these networks. The results are displayed in Figure~\ref{figsort} for some of the datasets, where the $x$-axis represents closeness centrality and the $y$-axis represents the reverse ranking of the nodes.

We study this curve in depth and find that the $4$-parameter logistic equation can better fit the curve. The equation of the $4$-parameter logistic is defined as,

\begin{center}
\begin{equation}\label{reverserank}
R_{rev}(u) = n + \frac{1-n}{1+\left( \frac{C(u)}{c_{mid}}\right) ^p},
\end{equation}
\end{center}
where, $c_{mid}$ represents closeness centrality of the middle ranked node in the network, $n$ represents total number of nodes, and $p$ denotes slope of the logistic curve at the middle point (also called hill's slope). All parameters are displayed in Figure~\ref{fig:sorted}. We will use this logistic equation to estimate the closeness rank of a node as explained in the next section.

\section{The Heuristic Method for Closeness Ranking}\label{sec-sam}

In this section, we propose a method to estimate the closeness rank of a node. The node of interest whose rank we compute is mentioned as \textit{interested node}. The proposed method exploits the structural characteristics of the closeness centrality to efficiently estimate the ranking. As we discussed, the reverse rank versus closeness centrality follows a sigmoid curve, we use this characteristic to estimate the rank of a node.
Once we estimate both parameters of the logistic Equation~\ref{reverserank}, the closeness rank of a node can be estimated in $O(1)$ time.

We will now discuss methods to estimate both of these parameters: ($1$) closeness centrality of the middle-ranked node ($c_{mid}$) and ($2$) slope of the logistic curve ($p$). Next, we will discuss the method to estimate closeness rank of a node and its complexity.

\begin{figure*}[htp]
  \centering
  \subcaptionbox{Brightkite}[.45\linewidth][c]{%
    \includegraphics[width=.45\linewidth]{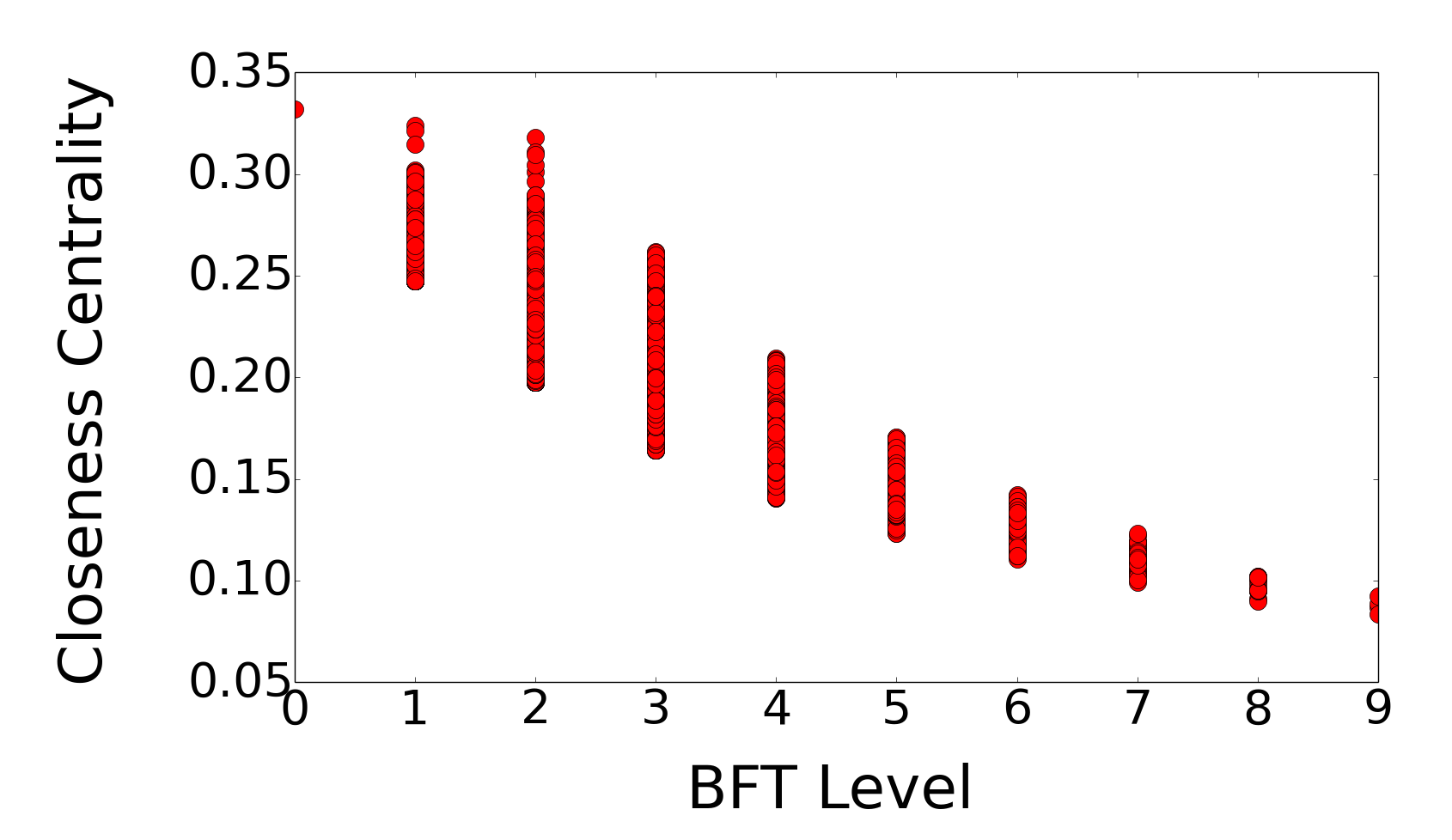}}\quad
  \subcaptionbox{DBLP}[.45\linewidth][c]{%
    \includegraphics[width=.45\linewidth]{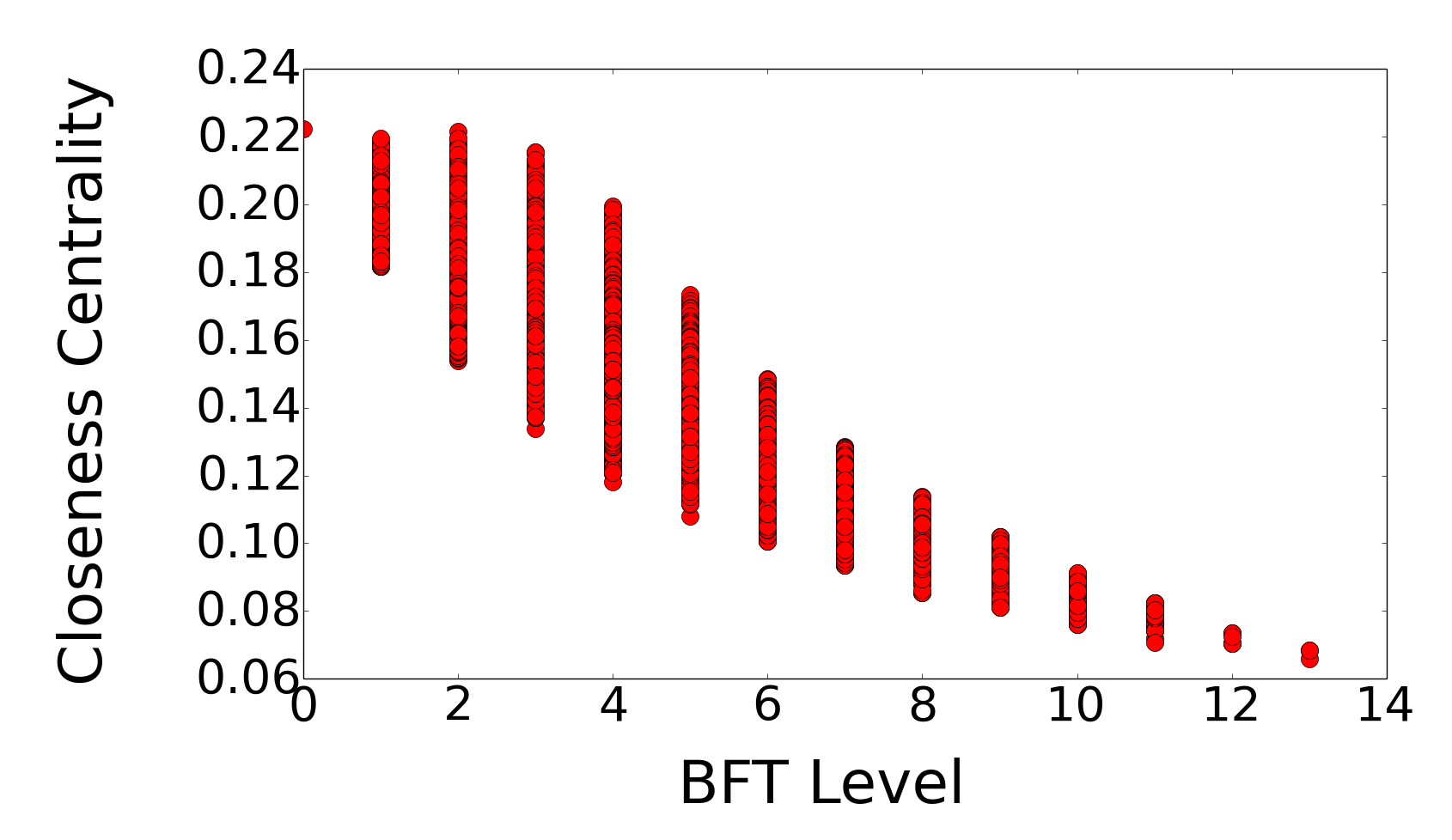}}\quad
  \subcaptionbox{Digg}[.45\linewidth][c]{%
    \includegraphics[width=.45\linewidth]{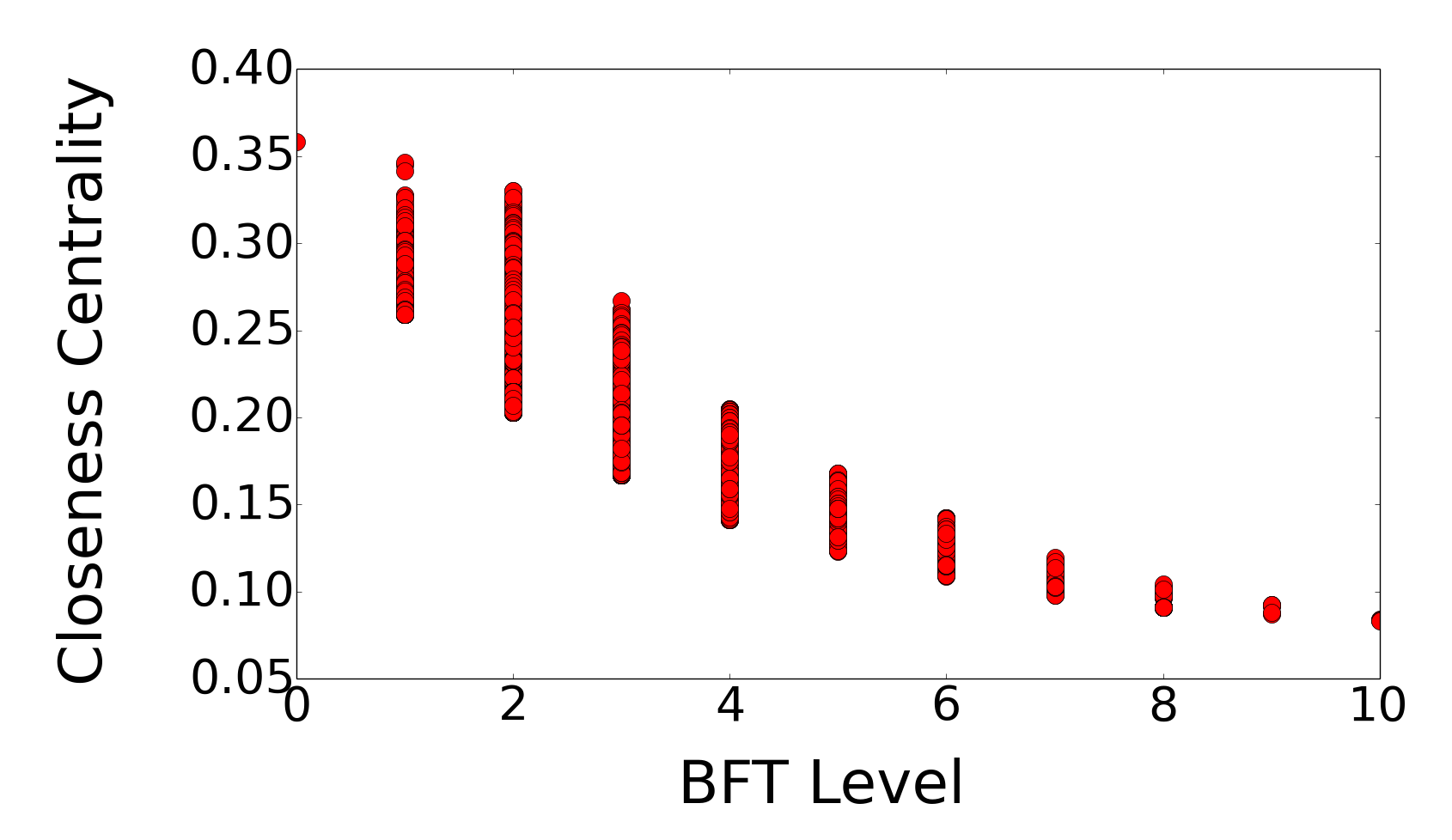}}\quad
  \subcaptionbox{Enron}[.45\linewidth][c]{%
    \includegraphics[width=.45\linewidth]{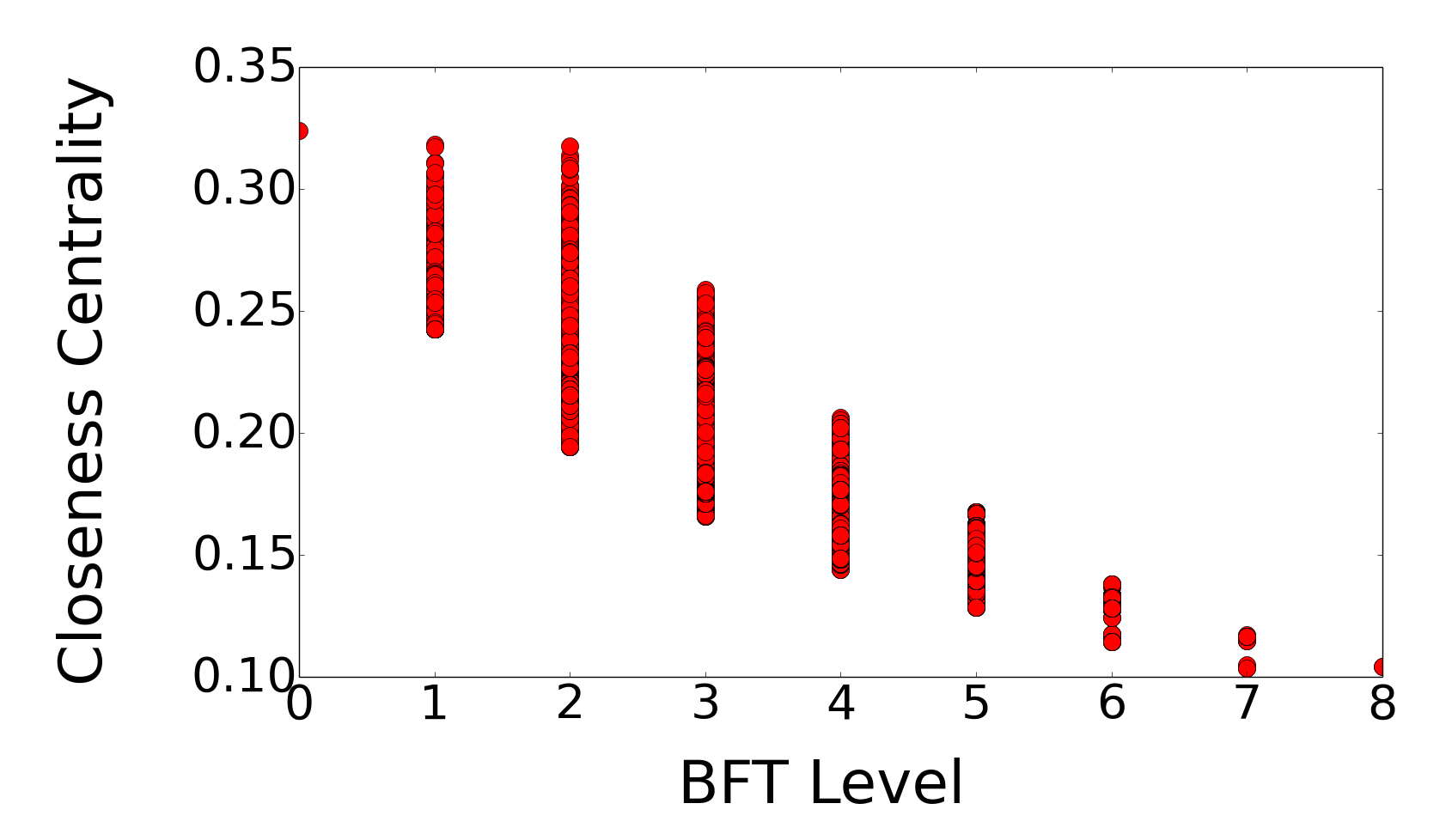}}\quad
  \subcaptionbox{Epinion}[.45\linewidth][c]{%
    \includegraphics[width=.45\linewidth]{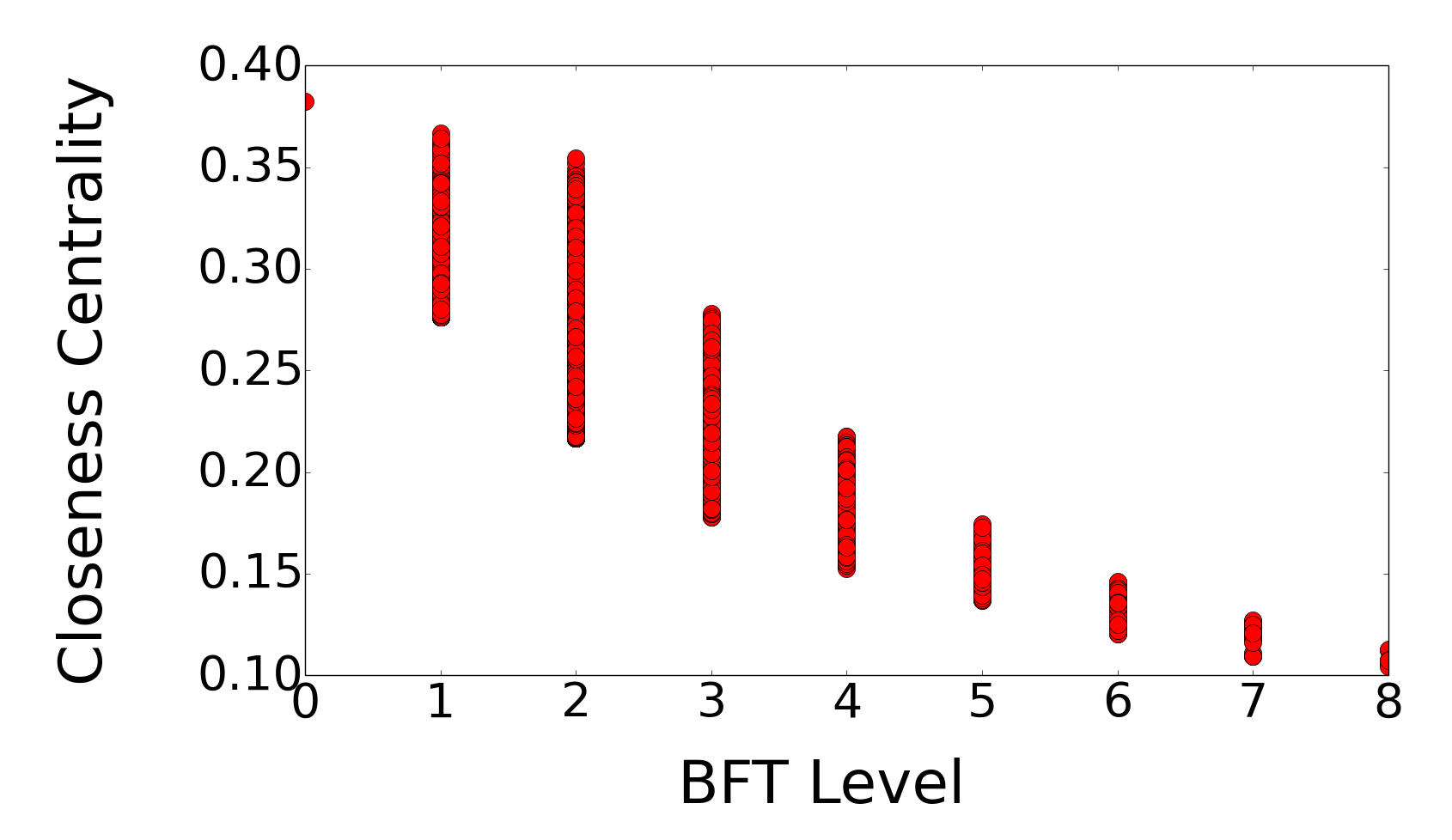}}\quad
  \subcaptionbox{Facebook}[.45\linewidth][c]{%
    \includegraphics[width=.45\linewidth]{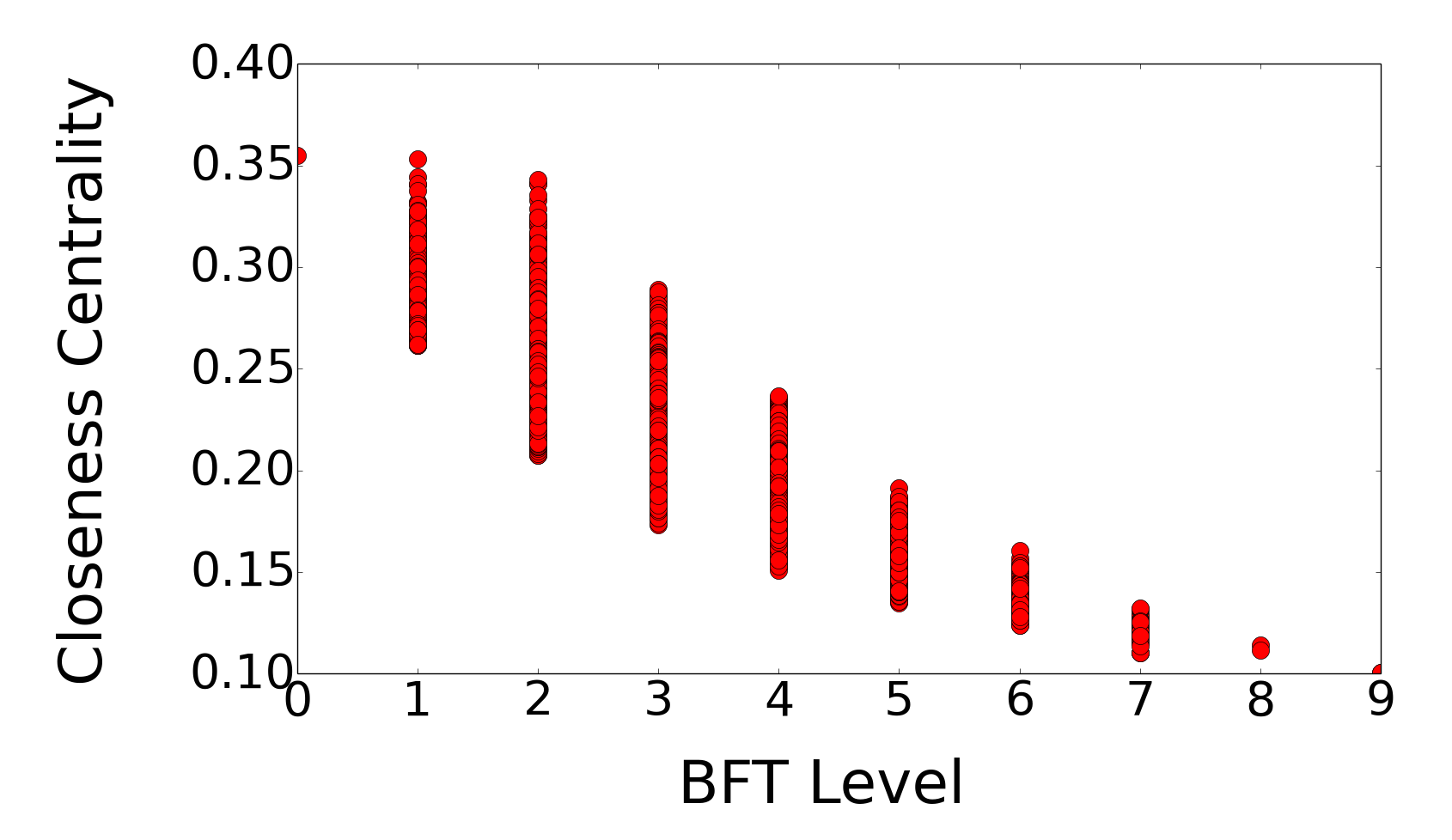}}\quad
  \caption{Closeness Centrality versus BFT Level from the Central Node}
  \label{fig-bft}
\end{figure*}

\subsection{Estimate closeness centrality of middle ranked node $(c_{mid})$}

We observed that in most of the real world networks, reverse rank versus closeness centrality follows a symmetric sigmoid pattern. The plots for some of these networks are shown in Figure~\ref{figsort}, where the plots $(a)-(d)$ are symmetric. We use this information to compute the value of $c_{mid}$, using the network.
 
Using the property \ref{sec-closedeg}, the maximum closeness centrality can be estimated as, $c'_{max}=C(u)$ where $deg(u) \geq deg(v), \forall v \in V$.
While estimating the closeness centrality of the \textit{interested node}, keep track of the node having the highest degree. Once the highest degree node is known, we can compute its closeness centrality using the standard computation method. We thus have the following observation.

\begin{observation}\label{obs1}
The maximum closeness centrality can be estimated as, $c'_{max}=C(u)$ where $deg(u) \geq deg(v),  \forall v \in V(G)$.
\end{observation}

In a network, the nodes having the maximum distance from the central node have minimum closeness centrality as observed in figure \ref{fig-bft}. So, we can keep track of the nodes falling on the outermost level of BFT while computing the maximum closeness centrality.
Let $w$ be a node in the network chosen uniformly at random from all the nodes farthest away from $u$ (i.e. $d(u,w)$ is maximum) for $u$ identified as a central node in Observation~\ref{obs1}.
Using the property \ref{sec-bft}, the minimum closeness centrality can be estimated by the closeness centrality of $w$. We thus have the following observation.

\begin{observation}
The minimum closeness centrality can be estimated as, $c'_{min}=C(w)$, $\exists w$ where $d(w, u)$ is max, for $u$ identified in Observation~\ref{obs1}.
\end{observation}
We now use Observations~$1$ and $2$ to estimate the closeness centrality of the middle ranked node.

\begin{proposition}
In the symmetric sigmoid curve of reverse rank versus closeness centrality, the closeness centrality of the middle-ranked node, $(c_{mid})$, can be computed as $c_{mid}=(c_{max}+c_{min})/2$.
\end{proposition}

\begin{proof}
If the sigmoid curve is symmetric, then using Figure~\ref{fig:sorted} we note that:
The distance from $C$ to $A$ = $c_{max} - c_{min}$ \\
The distance from $A$ to $B$ = $(c_{max} - c_{min})/2$ \\
The distance of $B$ from the origin point can be computed as,
\begin{equation}
c_{mid} = c_{min} + (c_{max}-c_{min})/2 = (c_{min}+c_{max})/2,
\end{equation}
as desired.~\end{proof}

\subsection{Estimate Slope of the Sigmoid Curve ($p$)}

We measured the slope of the sigmoid curve for $20$ real world networks using scaled levenberg-marquardt algorithm \cite{more1978levenberg} with 1000 iterations and 0.0001 tolerance. We observed that the slope ranges from 10-15. The slope for the discussed datasets is shown in Table~\ref{pvalue}. The average of these values is used as the value for $p$ in the simulation. We empirically observed that the slight variation in the estimation of $p$ does not cause more error in the ranking.

\begin{table}[h]
\centering
\caption{Networks versus their p values}
\label{pvalue}
\begin{tabular}{|l|l|l|l|}
\hline
Network  & p value & Network  & p value \\ \hline
Brightkit & 12.18   & Gowalla  & 10.79   \\ \hline
DBLP     & 14.11   & Google+    & 15.95   \\ \hline
Digg     & 14.23   & Facebook & 12.74   \\ \hline
Enron    & 11.47   & Twitter  & 14.47   \\ \hline
Epinion  & 12.99   & Slashdot & 14.89   \\ \hline
         &         & \textbf{Average}  & 13.38   \\ \hline
\end{tabular}
\end{table}


\subsection{Estimate Closeness Rank}



After estimating all the needed parameters for the sigmoid curve, the closeness rank of the \textit{interested node} $u$ can be estimated using Proposition~\ref{rank}.

\begin{proposition}\label{rank}
In a network $G$, the closeness rank of a node $u$ can be computed as, $R_{act}(u)=1 + \frac{n - 1}{1+\left(\frac{C_C(u)}{c_{mid}}\right)^p}$.
\end{proposition}

\begin{proof}

Using Equation \ref{reverserank}, the reverse rank of a node $u$ can be computed as,
\begin{center}
$R_{rev}(u)= n + \frac{1-n}{1+\left(\frac{C(u)}{c_{mid}}\right)^p}.$
\end{center}

The actual rank of a node can be estimated by subtracting its reverse rank from the total number of nodes plus $1$. We thus have that,

\begin{align*}
R_{act}(u)=n-R_{rev}(u)+1\\
R_{act}(u)=n- n - \frac{1-n}{1+\left(\frac{C(u)}{c_{mid}}\right)^p}+1,\\
\end{align*}
\begin{equation}\label{closerank}
R_{act}(u)=1 + \frac{n-1}{1+(\frac{C(u)}{c_{mid}})^p},
\end{equation}
as desired.~\end{proof}

Thus, we can now estimate the rank of a nodes in a network $G$ using Corollary~\ref{cor1}.

\begin{corollary}\label{cor1}
In a network $G$, the closeness rank of a node $u$ can be estimated as, $R_{est}(u)=1+\frac{n-1}{1+\left(\frac{C(u)}{c'_{mid}}\right)^{p'}}$, where $c'_{mid}$ and $p'$ are the estimated values of closeness centrality of middle ranked node and slope of the sigmoidal closeness centrality curve respectively.
\end{corollary}

The combined method to estimate the closeness rank of a node is explained in Algorithm~$1$. 
Here, $closeness\_centrality(G,u)$ method returns the closeness centrality of node $u$. $closeness\_centrality1(G,u)$ method returns the closeness centrality $C(u)$ of node $u$, the node $w$ having the highest degree in the network, and the network size $n$. $closeness\_centrality2(G,w)$ method returns the closeness centrality of node $w$, and list of the nodes having maximum distance from the node $w$. $random\_choice(cmin\_list)$ function returns a value uniformly at random from the given list $cmin\_list$. $closeness\_centrality(G,u)$, $closeness\_centrality1(G,u)$, and $closeness\_cent$- $rality2(G,u)$ methods can be implemented by modifying the BFT algorithm as they just need to keep track of few variables.



\begin{algorithm}[]
\caption{$EstimateClosenessRank(G,u,p)$}
$(C(u), w, n)= closeness\_centrality1(G,u)$\;
$(c'_{max}, cmin\_list) = closeness\_centrality2 (G, w)$\;
$c'_{min}= closeness\_centrality(G, random\_choice(cmin\_list))$\;
$c'_{mid}=(c'_{max}+c'_{min})/2$\;
Estimate closeness rank of the node using equation \ref{closerank} as, $R_{est}(u)=1 + \frac{n-1}{1+\left(\frac{C(u)}{c'_{mid}}\right)^p}$\;
Return $R_{est}(u)$\;
\end{algorithm}

\subsection{Complexity Analysis}
In this section, we will discuss the time complexity of the proposed heuristic method that is explained in Algorithm~$1$. The complexity of step 1 is $O(m)$ as it executes one BFT and keeps track of the highest degree node while executing the BFT. The time complexity of step 2 is $O(m)$ as it executes BFT from the node $w$ and returns the list of nodes that are traversed during the last level of BFT. The time complexity of step 3 is $O(m)$, as we assume that $random\_choice(cmin\_list)$ function returns a value in constant time as the size of the list is very small. Step 4 and 5 take $O(1)$ time. So, the overall complexity of the proposed method is $O(m) + O(m) + O(m) + 2\cdot O(1) = O(m)$. This is a great improvement over the classical ranking method that takes $O(n \cdot m)$ time.  


\section{The Randomized Heuristic Method}\label{sec-ram}

We observed that the accuracy of the estimated rank highly depends on the accuracy of the $c_{mid}$ estimator. So, we propose an improved $c_{mid}$ estimator to increase the accuracy of the proposed method.

In the heuristic method, we assumed that the sigmoid curve is symmetric. But in some real world networks, the sigmoid curve might not be symmetric, but very close as seen in Figure~\ref{figsort} $(e),(f)$, and the heuristic method will give a huge error for such cases. So, we propose an improved randomized method that uses uniformly random samples to estimate the value of $c_{mid}$. The improved method picks $k$ nodes uniformly at random and computes their closeness centrality values. The average of these $k$ closeness values is used as the estimated value of $c_{mid}$. Results show that the estimated value of $c_{mid}$ is very close to its actual value.
Further details are explained in the Results section. The complete algorithm is explained in Algorithm~$2$.


\begin{algorithm}[]
\caption{$RandomizedClosenessRank(G,u,p,k)$}
$(C(u), n)= closeness\_centrality3(G,u)$\;
Take a list $L$\;
\For{$i\leftarrow 1$ \KwTo $k$}{
Select a node $w$ uniformly at random\;
Add $closeness\_centrality(G,w)$ in list L\;
}
$c'_{mid}$ = average of all values of $L$\;
Estimate closeness rank of the node using equation \ref{closerank} as, $R_{est}(u)=1 + \frac{n-1}{1+\left(\frac{C(u)}{c'_{mid}}\right)^p}$\;
Return $R_{est}(u)$\;
\end{algorithm}

\subsection{Complexity Analysis}

In $RandomizedClosenessRank(G,u,p,k)$ algorithm, the complexity of $close$- $ness\_centrality3(G,u)$ method is $O(m)$ as it returns the closeness centrality of node $u$ and the total number of nodes. In the \textit{for loop}, $k$ nodes are chosen uniformly at random and their closeness centrality is computed.
So, the complexity of the for loop is $O(k \cdot m)$. Complexity of step 6 and 7 is $O(1)$. Thus, the overall complexity of the proposed method is $ O(m) + O(k \cdot m) + O(1) = O((k+1)m)$. As $k << n$, the complexity of the proposed method is $O(m)$.

\begin{figure*}[htp]
  \centering
  \subcaptionbox{DBLP}[.45\linewidth][c]{%
    \includegraphics[width=.45\linewidth]{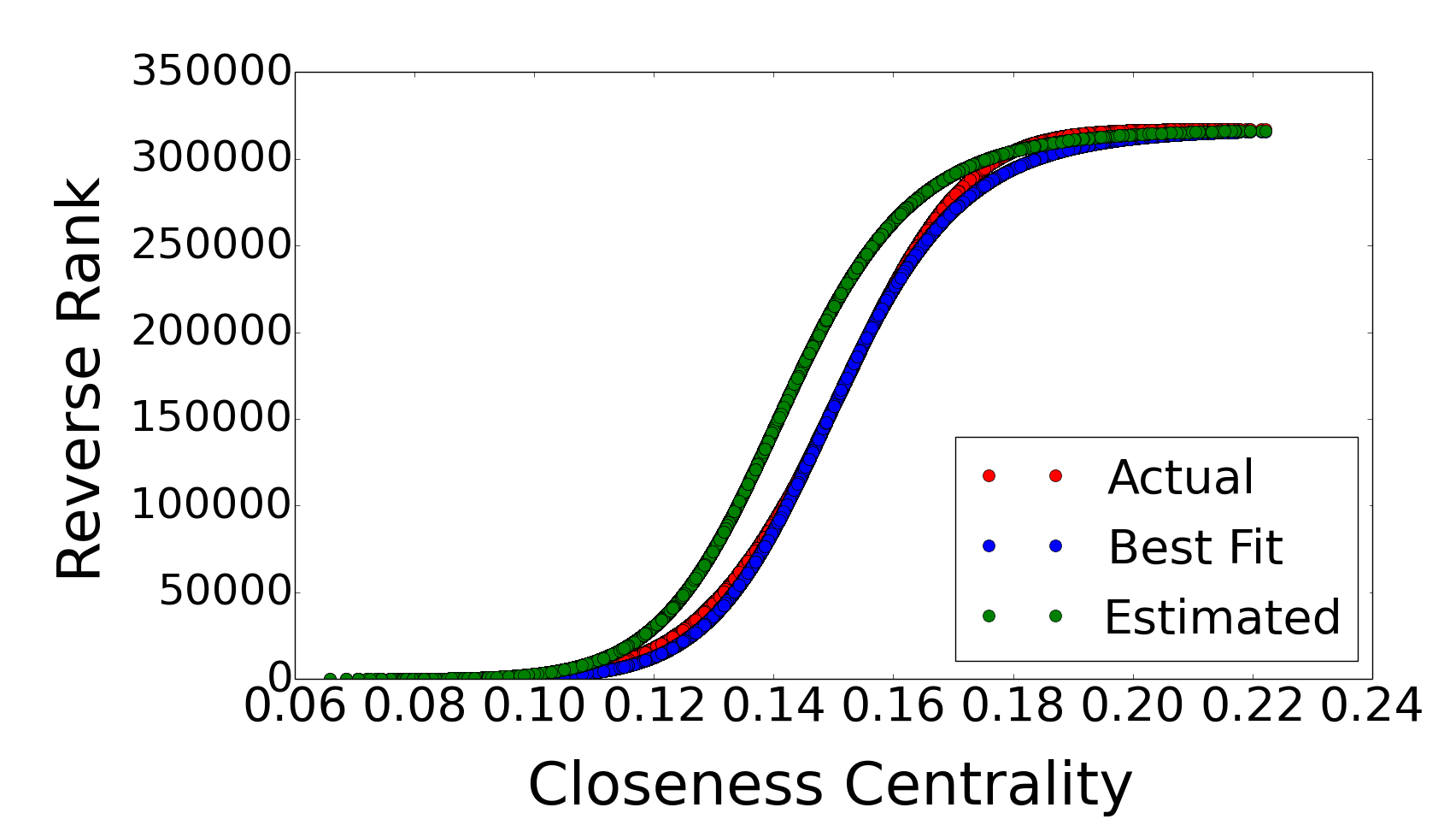}}\quad
  \subcaptionbox{Digg}[.45\linewidth][c]{%
    \includegraphics[width=.45\linewidth]{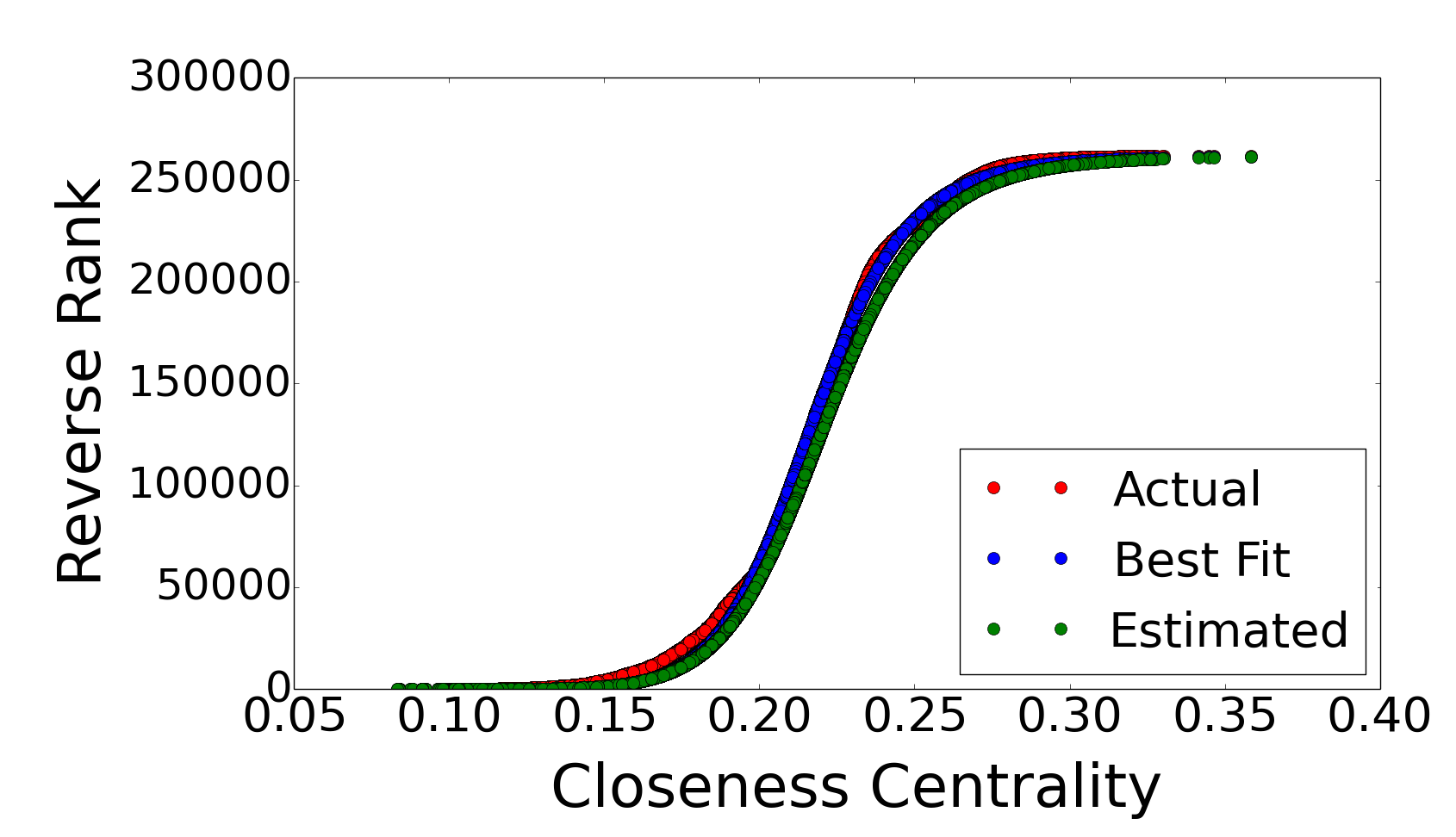}}\quad
  \subcaptionbox{Facebook}[.45\linewidth][c]{%
    \includegraphics[width=.45\linewidth]{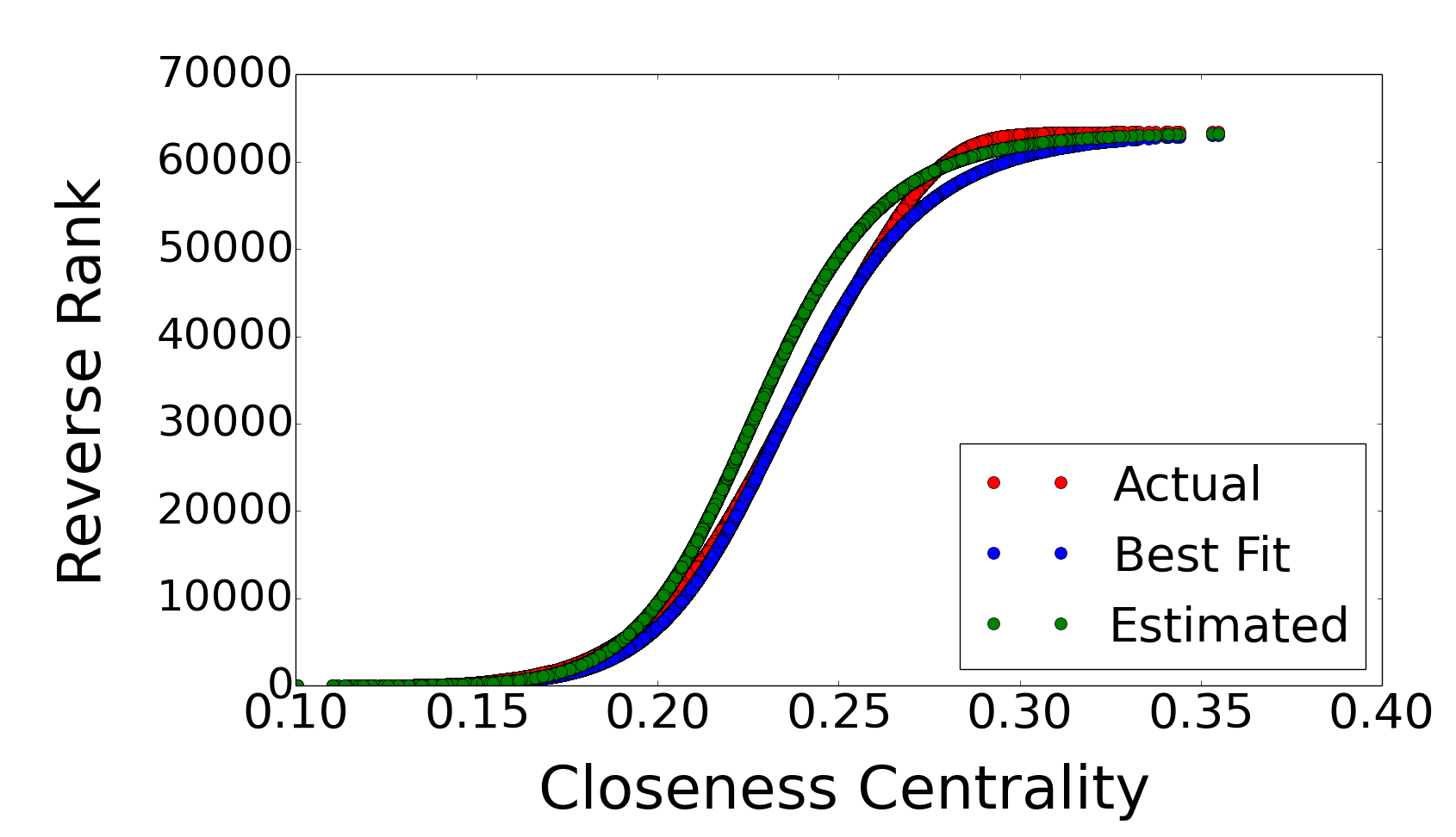}}\quad
  \subcaptionbox{Epinion}[.45\linewidth][c]{%
    \includegraphics[width=.45\linewidth]{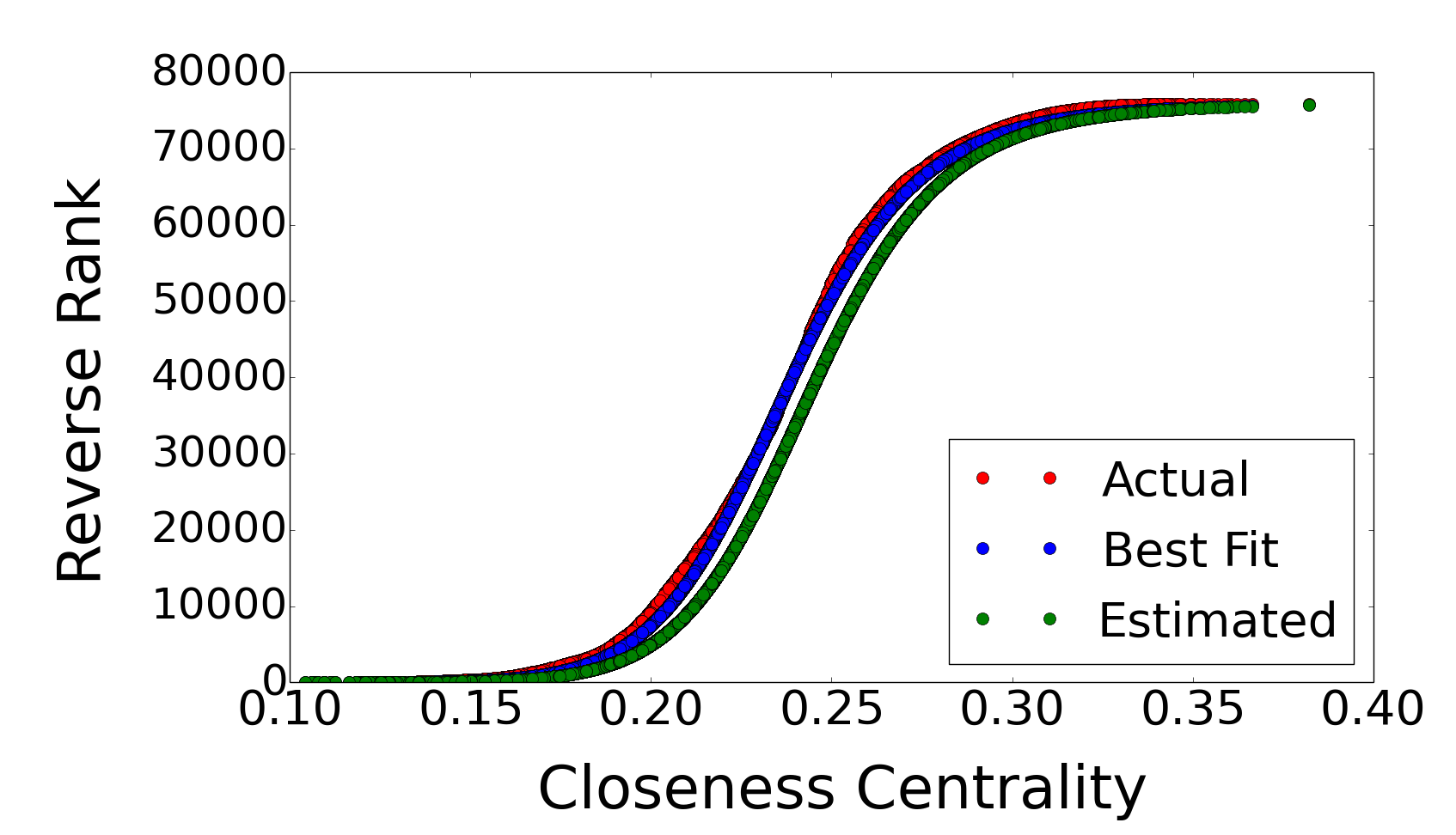}}\quad
  \subcaptionbox{Enron}[.45\linewidth][c]{%
    \includegraphics[width=.45\linewidth]{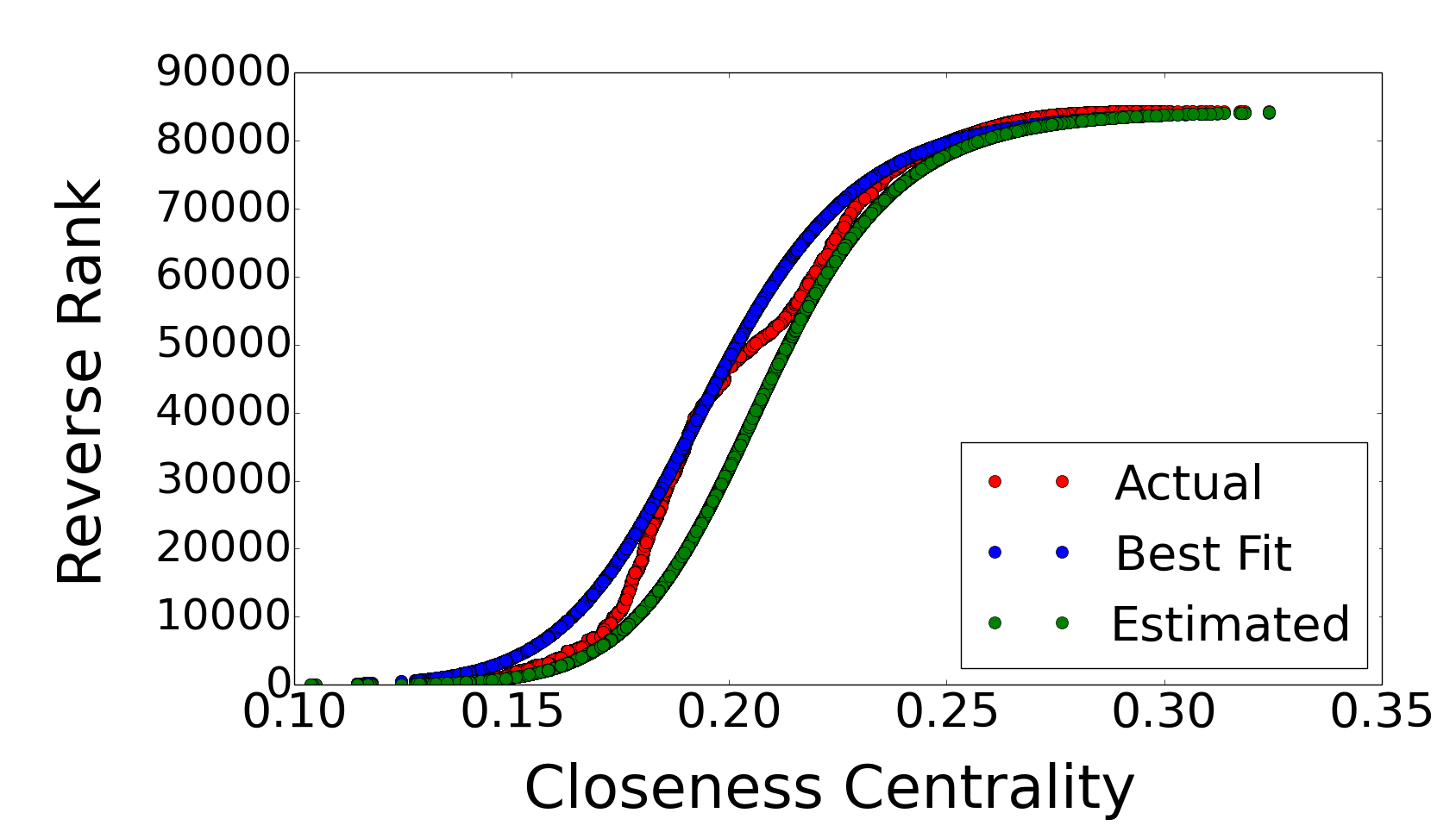}}\quad
  \subcaptionbox{Twitter}[.45\linewidth][c]{%
    \includegraphics[width=.45\linewidth]{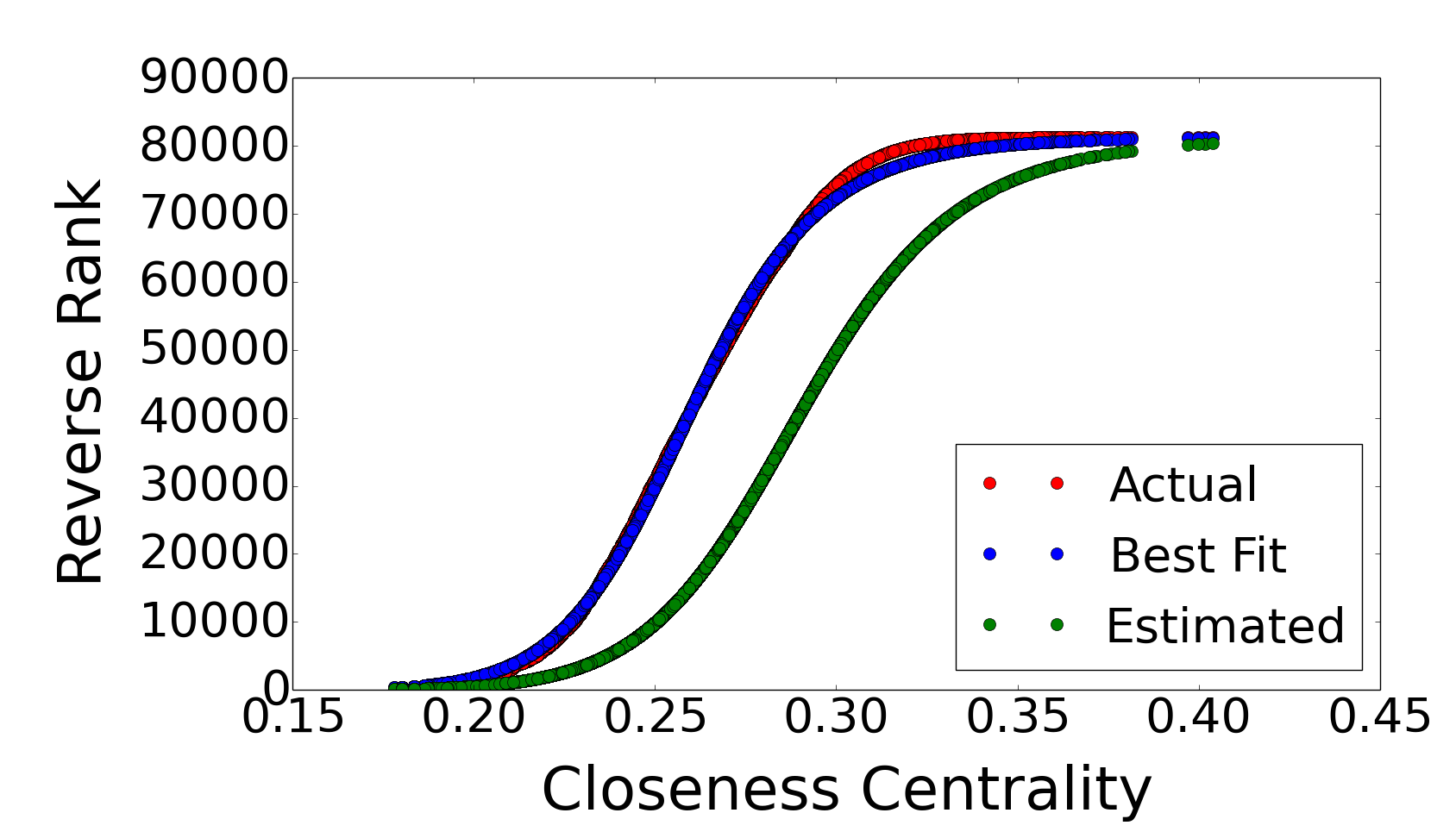}}\quad
  \caption{Reverse Rank versus Closeness Centrality}
  \label{figsort}
\end{figure*}

\section{Simulation Results}\label{sec-results}

In this section, we discuss error functions and simulation results on real world datasets.

\subsection{Error Functions}
The accuracy of the proposed methods is computed using absolute and weighted error functions that are discussed below:

\begin{enumerate}
\item \textbf{Absolute Error:} Absolute error for a node $u$ is computed as,
\begin{center}
$Err_{abs}(u) = |R_{est}(u) - R_{act}(u)|.$
\end{center}

The \textbf{percentage average absolute error} can be computed as $$ Err_{paae} = \frac{\text{average \; absolute \; error}}{\text{network \; size}} \cdot 100\%.$$

where the ${average \; absolute \; error}$ is computed by taking the average of absolute error for each node.

\item \textbf{Weighted Error:} In real life applications, importance of the rank depends on where does a node stand and how many total nodes are there. If you get 100 rank, it is admirable if there are 100,000 people but it is not good enough if there are only 500 people. So, the same rank value can be more important in a larger network than in a smaller one.
So, the impact of the error depends on percentile of the node as well as on the network size.
The proposed weighted error function considers both of these parameters and it is defined as,
\begin{center}
$Err_{wtd}(u) = \frac{Err_{abs}(u)}{n} \times percentile(u) \%.$
\end{center}

The percentile of a node $u$ can be calculated as, $percentile(u) = \frac{n - R_{act}(u) +1}{n} \times 100$. The weighted error increases linearly with the percentile of the node, and decreases with the network size.
\end{enumerate}

\subsection{Discussion}

\begin{table}[h]
\centering
\caption{Error in the Estimated Ranking}
\label{results}
\begin{tabular}{|l|l|l|l|l|l|l|}
\hline
Network & \multicolumn{2}{|c|}{Best Fit} & \multicolumn{2}{|c|}{Heuristic Method} & \multicolumn{2}{|c|}{Rand. Heuristic Method} \\ 

 & \multicolumn{2}{|c|}{Error} & \multicolumn{2}{|c|}{Error} & \multicolumn{2}{|c|}{Error (k=50)} \\ \hline


    & Paae & Wtd & Paae & Wtd & Paae & Wtd \\    \hline
    
Brightkite    &    1.48    &    0.79    &    7.96    &    3.89    &    2.96    &    1.32    \\    \hline
                                                        
DBLP    &    1.29    &    0.68    &    9.47    &    4.38    &    3.29    &    1.91    \\    \hline
                                                        
Digg    &    1.24    &    0.49    &    4.50    &    2.48    &    3.60    &    1.94    \\    \hline
                                                        
Enron    &    4.14    &    1.81    &    8.56    &    3.92    &    5.64    &    2.94    \\    \hline
                                                        
Epinion    &    1.42    &    0.70    &    7.72    &    3.96    &    2.16    &    1.00    \\    \hline
                                                        
Facebook    &    1.72    &    0.94    &    6.68    &    3.54    &    3.64    &    1.92    \\    \hline
                                                        
Google+    &    1.44    &    0.64    &    20.57    &    12.22    &    3.70    &    1.77    \\    \hline
Gowalla    &    1.52    &    0.81    &    17.96    &    9.30    &    3.97    &    1.94    \\    \hline
Slashdot    &    0.78    &    0.29    &    6.29    &    2.69    &    2.91    &    1.47    \\    \hline
                                                        
Twitter    &    0.98    &    0.59    &    24.64    &    14.70    &    3.41    &    1.96    \\    \hline

\end{tabular}                                                                        
\end{table}

The proposed methods are simulated on all datasets discussed in Table~\ref{datasets}. To measure the accuracy of the proposed methods, the absolute and weighted errors are computed for each node. Then it is averaged over all nodes to compute the overall error in the proposed methods. The errors of the proposed methods are shown in Table~\ref{results}.

The best fit error is computed by using the best-fit logistic curve on the reverse rank versus closeness centrality curve. The parameters of the best fit curve are computed using scaled levenberg-marquardt method with 1000 iterations and 0.0001 tolerance~\cite{more1978levenberg}. Once the parameters of the best fit curve are computed, the Equation~\ref{closerank} is used to compute the rank of a node. Results show that the error computed using best-fit parameters is very low, and the sigmoid closeness pattern can be efficiently used to estimate the rank of the nodes.

The error for heuristic and randomized heuristic methods is shown in Table~\ref{results} for $p=13.38$. The error varies with the $p$ value. The reverse rank versus closeness centrality plots for the best fit and approximated parameters are shown in Figure~\ref{figsort}.
The heuristic method gives a high error on some of the real world networks due to the error in the estimation of the parameters of the logistic curve or if the curve is not smooth.
The complexity to estimate the closeness rank of a node is same as computing its closeness centrality; this is a great improvement over the classical ranking method.

Next, we show that the improved randomized heuristic method gives a great improvement over the heuristic method. To compute the error, each experiment is repeated $40$ times for $k=50$, and the average of the errors is shown in Table~\ref{results}.

Thus, the results show that the sigmoid behavior of closeness centrality can be used to fast estimate the closeness rank of a node. The proposed methods can be efficiently applied on real world online networks as the proposed methods do not need to store the network. Their APIs can be used to run the BFT and compute the closeness centrality.

\section{Conclusion}\label{sec-conclusion}

In the present work, we studied behavior of the closeness centrality and its correlation with the structural properties of the network on real world scale-free social networks. We observed that the reverse ranking versus closeness centrality follows a sigmoid pattern. We further analyze how does the closeness centrality of nodes change as we move from the central region to the periphery. These unique characteristics of closeness centrality are used to propose heuristic method for closeness ranking of a node. The complexity of the proposed method is $O(m)$ that is a great improvement over the classical ranking method that takes $O(n \cdot m)$ time.

The proposed method is further improved using uniformly random samples where the closeness centrality values of $k$ sampled nodes are used to estimate a parameter $(c_{mid})$ of the sigmoid curve. The complexity of the improved method is $O(k \cdot m) \approx O(m)$ as ($k << n$). The accuracy of the proposed methods is verified using absolute and weighted error functions. Results show that the proposed methods can be efficiently used to estimate the closeness rank of a node.


\section{Future Directions}\label{sec-future}


In the proposed methods, we have estimated the slope of the sigmoid curve as an average of the slopes observed in real world networks. The slope of the curve denotes how sharply the closeness centrality increases for all middle layered nodes from periphery to central region. It depends on the density of the network and how the density changes from periphery to the center.

We study the correlation of the slope of the sigmoid curve ($p$) with the density of the network. In figure~\ref{density1}, we plot slope versus density for BA networks~\cite{barabasi1999emergence} having $40,000$ nodes.
We observe that the slope increases with the density, reaches its maximum value, and then it further decreases with the increase in the density. This correlation can be used to propose an estimator for $p$ value.

\begin{figure}[htp]
\centering
\includegraphics[width=0.8\linewidth]{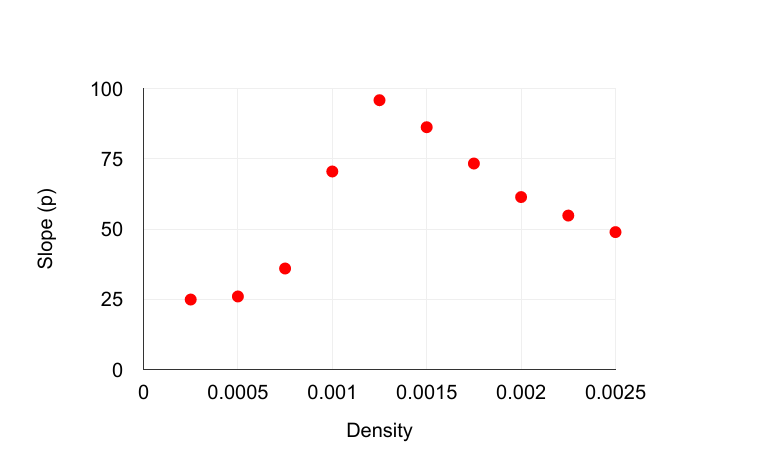}
\caption{Slope vs. Density on BA networks.}
\label{density1}
\end{figure}

The proposed methods can be further improved by estimating the closeness centrality of a node using its local information without having the entire network. These methods will be highly required due to the fast growth of the networks. This will improve the time complexity of the proposed methods many folds. The proposed methods can also be extended to other types of networks like weighted networks, directed networks, hypergraphs, and so on.

The rank estimation of a node based on other centrality measures like betweenness centrality, coreness, PageRank, is still an open problem.

\bibliographystyle{elsarticle-num}

\bibliography{/home/akrati/latex/mybib}

%
%
%
%

\end{document}